\documentclass[runningheads]{llncs}
\usepackage{amsmath,amssymb,amsfonts,mathtools}
\usepackage{algorithm}
\usepackage{algpseudocode}
\usepackage{hyperref}
\usepackage{braket}
\usepackage{bm}
\usepackage{graphicx}
\usepackage{float}
\usepackage{thmtools}
\usepackage{thm-restate}
\usepackage{capt-of}
\usepackage[font=small,labelfont=bf,tableposition=top]{caption}
\DeclareCaptionLabelFormat{andtable}{#1~#2  \&  \tablename~\thetable}

\usepackage[T1]{fontenc}

\usepackage{booktabs} %
\usepackage{xcolor}   
\usepackage{siunitx}  

\definecolor{lightgray}{gray}{0.95}

\usepackage{amsmath,amsfonts,amssymb}

\usepackage{array} 
\usepackage{tabularray}
\usepackage{nicematrix}
\usepackage{booktabs}
\usepackage{colortbl}
\usepackage{ragged2e}
\usepackage{longtable}

\usepackage{multirow}
\usepackage{hhline}
\usepackage{makecell}
\usepackage{rotating}
\usepackage{braket}

\usepackage[caption=false,font=normalsize,labelfont=sf,textfont=sf]{subfig}
\usepackage{textcomp}
\usepackage{stfloats}

\usepackage{url}
\usepackage{verbatim}

\usepackage{cite}

\usepackage{microtype}

\usepackage{pgf,tikz}

\usepackage{tikz}
\usepackage{grffile}
\usepackage{circuitikz}
\usetikzlibrary{arrows,mindmap,decorations,shapes,shapes.geometric,positioning,arrows.meta,shadows,trees,fit,calc}

\usepackage{xspace}
\definecolor{joli}{RGB}{225,95,0}
\definecolor{JOLI}{RGB}{225,95,0}

\makeatletter
\newcount\tikzcountgrandchild
\def\tikz@grow@concentric{%
  \pgftransformreset%
  \pgftransformrotate{(\pgfkeysvalueof{/tikz/sibling angle})*(\tikzcountgrandchild)}%
  \ifnum\tikztreelevel=1
    \pgftransformrotate{(\pgfkeysvalueof{/tikz/sibling angle})*(\pgfkeysvalueof{/tikz/noc}-1)/2}%
  \fi
  \pgftransformxshift{\the\tikzleveldistance}%
  \ifnum\tikztreelevel=2
    \global\advance\tikzcountgrandchild by1
 \fi
}
\tikzset{
    noc/.initial=0,
    branch color/.style={
        concept color=#1!white,
        every child/.append style={concept color=#1!white!30!white},
    }
}

\usepackage[utf8]{inputenc}
\title{\textsc{SILMARILS}: Information-Theoretic and Quantum-Secure Designated-Verifier Signatures}
\titlerunning{\textsc{SILMARILS}: An IT-Secure Signature Scheme Resilient to Quantum Attacks}
\author{Hassan Khodaiemehr\inst{1} \and Khadijeh Bagheri\inst{1} \and Chen Feng\inst{1} \and Dariia Porechna\inst{2} }
\institute{Faculty of Applied Science, School of Engineering, The University of British Columbia (UBC)
Okanagan Campus, Kelowna, BC, Canada\\
\{hassan.khodaiemehr,khadijeh.bagheri,chen.feng\}@ubc.ca \and EternaX Labs\\
dariia.p@eternax.ai}

\begin{document}
\maketitle
\begin{abstract}
\textsc{SILMARILS} follows the principle of \underline{S}ecure \underline{I}nformation-theory-\underline{L}everaged
\underline{M}echanism for \underline{A}uthentication and
\underline{R}eceipt-\underline{I}ntegrated \underline{L}ightweight \underline{S}ignatures.
It is built from a minimal algebraic core over $\mathbb{F}_p$ using true randomness and
perfect $2$-out-of-$2$ Shamir secret sharing. The framework supports both two-party and
three-party modes.
In the two-party setting, \textsc{SILMARILS} realizes a transferable designated-verifier
(TDV) signature scheme. The designated verifier can simulate accepting transcripts
indistinguishable from real ones, achieving Jakobsson--Sako--Impagliazzo DV security.
The verifier may publish a receipt $r$ enabling public verification, yet even with $r$, no
external party can tell whether a transcript was signed or simulated. As DV signatures
permit simulation, standard EUF-CMA cannot hold for the designated verifier; instead, we
prove $\mathsf{EUF\text{-}CMA}^{\neg\mathsf{DV}}$ security for all non-designated verifiers
in both the random oracle model (ROM) and quantum random oracle model (QROM).
In the three-party mode, adopting the broadcast model of
Fitzi~\emph{et~al.}~\cite{FitziWolf2004,Fitzi2001,IT-book}, we obtain a statistically secure
signature protocol with simulation-based security and error~$1/p$. We analyze security in
the Pure IT model, the IT+ROM, and the QROM, extending the Fitzi~\emph{et~al.} framework
to quantum adversaries with classical I/O. Correctness, secrecy, transferability, and
unforgeability for non-designated parties remain equivalent to simulation-based security.
Thanks to its simple algebraic structure, \textsc{SILMARILS} offers very compact
keys and signatures for the blockchain settings we target, where standardized
PQC schemes are already more than sufficient. Our goal is not to compare
\textsc{SILMARILS} with PQC, but to highlight its suitability for lightweight TDV
authentication. A fair comparison with other DV schemes is omitted due to space
and the complexity of aligning models.
\end{abstract}

\section{Introduction}
Digital signatures ensure authenticity in distributed systems, with constructions spanning classical schemes \cite{RFC3447,FIPS1861,ECDSA,Boneh2005Schnorr} and post-quantum alternatives \cite{Lamport1979,merkle1989,SPHINCSplus,CourtoisFiniaszSendrier2001,McEliece_ref,dilithium2020,falcon2018,uov2025spec,cryptoeprint:2022/214,sqisign2025spec}. Information-theoretic (IT) signatures provide unconditional security \cite{ChaumRoijakkers1991,Hanaoka2000,AmiriEtAl2018}, but current schemes often lack full simulation-based guarantees or remain inefficient.

The study of digital signatures is deeply intertwined with the theory of
broadcast and Byzantine agreement \cite{IT-book}. The foundational works of Pease, Shostak,
and Lamport~\cite{PeaseShostakLamport1980,LamportShostakPease1982} introduced
the broadcast and Byzantine agreement problems and established both feasibility
and impossibility results under various adversarial models. The attack-coupling
argument of Fischer, Lynch, and Merritt~\cite{FischerLynchMerritt1986} has since
become a standard tool for proving impossibility results in distributed
computing~\cite{Considine2005,FitziGarayMaurer2005}. A comprehensive
treatment of distributed algorithms can be found in~\cite{Lynch1996}.
Digital signatures play a crucial role in enabling authenticated broadcast.
Dolev and Strong~\cite{DolevStrong1979} showed that broadcast is possible for
any number of dishonest parties if signatures are available along arbitrary
transfer paths. Conversely, signatures can be constructed from broadcast via
information-checking protocols~\cite{CramerDamgaard1999,CramerDamgaardNielsen2015,IT-book}.
The minimal requirements for broadcast were further studied in
\cite{FitziGarayMaurer2005}. Randomized and variable-round protocols were
introduced to overcome deterministic lower bounds~\cite{BenOr1983,Rabin1983},
with optimal expected-round protocols developed in
\cite{FeldmanMicali1997,KatzKoo2009}. In asynchronous networks, impossibility
results such as~\cite{FischerLynchPaterson1985} motivated randomized protocols
for consensus~\cite{Aspnes2003}.

\subsection{Motivation and Contribution}
Although IT digital signatures are known to be
possible, existing constructions do not provide efficient, reusable, and
simulation-secure signatures in the multi-party, multi-use setting. The
central challenge is achieving unconditional security while retaining
efficiency, scalability, and practical deployability. Our work is
motivated by the intersection of several research directions:
IT security, secure multiparty computation (MPC),
long-term quantum resilience, and the authentication requirements of
modern distributed systems such as blockchains. In particular,
\textsc{SILMARILS} addresses two complementary settings: a two-party
TDV mode that yields efficient
$\mathsf{EUF\text{-}CMA}^{\neg\mathsf{DV}}$ signatures in the ROM and QROM, and a three-party
broadcast-model construction that achieves full simulation-based
security in the sense of Fitzi~\emph{et~al.} notion.

IT signatures are appealing in MPC settings, where using computational signatures would break end-to-end IT security. In practice, these guarantees are instantiated via hash functions, either as injective encodings in the Pure IT model or as random oracles in the IT+ROM and QROM settings, which requires analyzing the algebraic core in the ROM/QROM. Thus, unlike post-quantum signatures that rely on hardness assumptions (e.g., LWE or hash-based security),
our two-party construction achieves computational security in the ROM/QROM using IT tools, while the
three-party construction achieves IT security in the Pure IT model (replacing hashes with IT authentication codes).
In the three-party setting with at most one dishonest party, the classical lower bound shows that any \(\varepsilon\)-secure digital signature \emph{from scratch} which uses only noiseless channels and no setup, must satisfy \(\varepsilon \ge 1/3\)~\cite[Corollary~18.7]{IT-book}. The \(\varepsilon = 1/p\) construction of~\cite[Theorem~18.18]{IT-book} avoids this bound only by working in a richer setup with authenticated channels, pre-shared keys, or correlated randomness. This motivates our move to the TDV framework. Although TDV signatures superficially resemble MACs due to shared randomness, they are fundamentally different: MACs are symmetric, two-party, and non-transferable, whereas TDV schemes are publicly readable and verifiable only by a designated verifier who may choose to reveal validity evidences. In our two-party setting, the shared randomness is ephemeral rather than a
long-term secret key, and is used solely to enable controlled transferability
of signatures rather than to authenticate messages. This distinction ensures that our model and impossibility results remain fully consistent with the broader landscape of three-party signature constructions.
Hence, our goal is to develop an information-theoretic signature framework that is
efficient, reusable, and compatible with both classical and quantum settings.
\textsc{SILMARILS} achieves this by combining a minimal algebraic core over
$\mathbb{F}_p$ with perfect $2$-out-of-$2$ Shamir secret sharing. The framework
supports two complementary modes.

\textbf{Two-party mode:}
We obtain a transferable designated-verifier (TDV) signature scheme with
computational designated-verifier simulatability in the sense of
Jakobsson--Sako--Impagliazzo and $\mathsf{EUF\text{-}CMA}^{\neg\mathsf{DV}}$
security for all non-designated verifiers in both the ROM and QROM. This mode
provides lightweight authentication suitable for user--validator interactions
and other two-party settings where public verifiability is not required.

\textbf{Three-party mode:}
Adopting the broadcast model of Fitzi~\emph{et~al.}, we construct a
statistically secure signature protocol achieving full simulation-based
security with error~$1/p$. Because practical instantiations rely on hash
functions, we analyze security in the Pure IT, IT+ROM, and QROM models and
extend the Fitzi simulation framework to a quantum setting with classical I/O.

These two modes together capture the full range of practical and theoretical
requirements: efficient two-party authentication, and a three-party
information-checking (IC) structure supporting correctness, unforgeability,
transferability, and secrecy. Table~\ref{tab:signature-families} summarizes the comparison of \textsc{SILMARILS} and other signature families.
\begin{table}[tb]
\centering
\caption{Comparison of major digital signature families.}
\label{tab:signature-families}
\resizebox{0.99\textwidth}{!}{
\begin{tabular}{@{}l l l l l@{}} 
\toprule
\textbf{Family} & \textbf{Assumption} & \textbf{Key Sizes} & \textbf{Signature Sizes} & \textbf{Quantum Security} \\
\midrule
Classical (RSA, ECDSA) & Factoring, DL, ECDL & Small--Medium & Small & No \\
Hash-Based (SPHINCS$^+$) & Hash functions & Small & Large & Yes \\
Code-Based (CFS) & Decoding random codes & Large & Medium & Yes \\
Lattice-Based (Dilithium, Falcon) & M-SIS/M-LWE, SIS over NTRU Lattices & Medium & Small--Medium & Yes \\
Multivariate (UOV) & MQ hardness & Large & Small & Yes \\
Isogeny-Based (SQISign) & Isogeny problems & Small & Small & Yes (emerging) \\
\rowcolor{lightgray}
\textbf{IT-Based (\textsc{SILMARILS})} & \textbf{Signer–Verifier shared randomness (TDV)} & \textbf{Very Small} & \textbf{Very small} & \textbf{Yes} \\
\bottomrule
\end{tabular}
}\vspace{-0.7cm}
\end{table}

\subsection{Applications to PQ Blockchain Architectures}
The properties of \textsc{SILMARILS} align naturally with the needs of blockchain systems. Constant-time verification, compact receipts, and IT security allow PQ-safe architectures without the heavy onchain footprint of standardized PQ signatures. We developed a PQ-safe blockchain protocol where \textsc{SILMARILS} replaces both validator signatures during consensus and user transaction message authentication, demonstrating that IT techniques can deliver practical, high-performance authentication in both classical and post-quantum distributed systems. The key achievements of our architecture include: 1) The TLS-derived per-pair secrets eliminate the need for large user-facing PQ signatures in high-frequency interactions, while designated-verifier receipts preserve public auditability. 2) Validator-to-validator authentication becomes cheaper, consensus protocols simplify, and randomness generation can be implemented via lightweight commit–reveal mechanisms. The concrete ledger integration is outside the scope of this paper since we intend a companion publication.

The rest of this paper is organized as follows.
Section~\ref{Pre_sec} presents the necessary preliminaries, including
notation, ideal functionalities, and the security framework used
throughout the paper.
Section~\ref{sec:quantum-lemma-18-15} introduces the quantum-information
tools used in our analysis.
Section~\ref{Construction_sec} describes our digital signature scheme
\textsc{SILMARILS} in its two-party and three-party modes, detailing the
key generation, signing, and verification procedures.
Section~\ref{sec:dv-security} and Section~\ref{Security_sec}  provide a rigorous
security analysis of the two-party mode of \textsc{SILMARILS} under both the
Jakobsson--Sako--Impagliazzo designated-verifier framework and the
$\mathsf{EUF\text{-}CMA}^{\neg\mathsf{DV}}$ unforgeability notion, while
Section~\ref{Security_sec2} analyzes the security of the three-party mode.
Section~\ref{sec:performance} evaluates the efficiency of our scheme.
Finally, Section~\ref{sec:conclusion} concludes the paper.
\section{Preliminaries}\label{Pre_sec}
Let $p$ be a prime and let $\mathbb{F}_p$ denote the finite field of order $p$.
We write $\mathbb{F}_p^\ast = \mathbb{F}_p\setminus\{0\}$ for its multiplicative
group. Sampling uniformly from $\mathbb{F}_p^\ast$ is denoted by
$x \xleftarrow{\$} \mathbb{F}_p^\ast$.
For a random variable $X$ over a finite set $\mathcal{X}$, we write
$P_X$ for its distribution. For two distributions $P$ and $Q$ over
$\mathcal{X}$, the total variation distance is
\begin{eqnarray}\label{d_var}
 d_{\mathrm{var}}(P,Q)
&=& \frac12 \sum_{x\in\mathcal{X}} |P(x)-Q(x)|
= \max_{A\subseteq\mathcal{X}} |P(A)-Q(A)|.
\end{eqnarray}
A function $\mu:\mathbb{N}\to[0,1]$ is negligible if for every polynomial
$q(\cdot)$ there exists $k_0\in \mathbb{N}$ such that $\mu(k) \le 1/q(k)$ for all $k\ge k_0$.
We use boldface to denote vectors and tuples when needed, and we use
$\mathsf{View}_i$ to denote the view of party $\mathcal{P}_i$ in a protocol execution.
We use the perfect $2$-out-of-$2$ Shamir secret sharing scheme
$\mathsf{SSS}_{2,2}$ over $\mathbb{F}_p$ \cite{Shamir_sec}. Given a secret $s\in\mathbb{F}_p$, the
dealer chooses a random slope $a \xleftarrow{\$} \mathbb{F}_p$ and defines the
degree-$1$ polynomial
\(
f(x) = s + ax
\).
Let $w_0,w_1\in\mathbb{F}_p$ be fixed, distinct, public interpolation points or \emph{weights}.
The shares are $s_0 = f(w_0)$ and $s_1 = f(w_1)$.
Now, given the shares $(s_0,s_1)=\mathsf{SSS}_{2,2}(s)$, and public weights $w_0,w_1$, reconstruction function $\mathsf{SSS}_{2,2}^{-1}$ is defined by Lagrange interpolation:
\[
s =\mathsf{SSS}_{2,2}^{-1}(s_0,s_1)= \frac{w_0 s_1 - w_1 s_0}{w_0 - w_1}.
\]
Each individual share is uniformly distributed over $\mathbb{F}_p$ and
independent of the secret $s$, so $\mathsf{SSS}_{2,2}$ is a perfect
secret-sharing scheme~\cite{Shamir_sec}. Since its security is purely
information-theoretic, Shamir's Secret Sharing enjoys a natural form of
quantum resilience: its guarantees do not rely on computational
assumptions that could be broken by quantum algorithms, and remain valid
even in the presence of quantum adversaries~\cite{shamir_quantumsafe}.
\subsection{Ideal Functionality for Three-Party IT-Based Digital Signatures}\label{subsec:model}
We analyze our construction in the standard synchronous MPC model~\cite{IT-book}, the minimal setting supporting IT-based authentication~\cite{Access_paper}. The system consists of three parties: the signer $\mathcal{P}_1$, the holder $\mathcal{P}_2$, and the verifier $\mathcal{P}_3$.
Parties communicate over pairwise private and authenticated channels and have access to authenticated broadcast. Pairwise channels ensure confidentiality and integrity of point-to-point messages, while broadcast prevents equivocation. This matches the communication assumptions of the three-party \textsc{SILMARILS} protocol and the IC-based framework of~\cite[Protocol~18.7]{IT-book}.
We consider a static, active adversary corrupting at most one party. Upon corruption, the adversary obtains full control of that party and may arbitrarily deviate from the protocol. The one-corruption threshold is the minimal requirement for IT-based authentication in the three-party setting~\cite{IT-book}.
\begin{definition}
For a protocol $\Pi$ and adversary $\mathcal{A}$ corrupting party
$\mathcal{P}_i$, the view $\mathsf{View}_i^\Pi$ consists of: 1) the input of $\mathcal{P}_i$, 2) the private randomness of $\mathcal{P}_i$, 3) all messages received by $\mathcal{P}_i$, 4) all oracle answers observed by $\mathcal{P}_i$.
\end{definition}
\begin{figure}[b]
\vspace{-0.6cm}
\centering
\resizebox{0.7\textwidth}{!}{
\begin{tikzpicture}[>=stealth,thick]
\node at (1.5,1.0) {$\mathcal{P}_1$};
\node at (6,1.0) {$\mathcal{P}_2$};
\node at (11,1.0) {$\mathcal{P}_3$};

\node (f1) [draw, rounded corners, minimum width=2.4cm, minimum height=1cm]
      at (3,0) {$f_{\mathrm{DS1}}$};

\node (f2) [draw, rounded corners, minimum width=2.4cm, minimum height=1cm]
      at (9,0) {$f_{\mathrm{DS2}}$};

\node (decision) [rounded corners, minimum width=4.4cm, minimum height=1.4cm]
      at (13.5,0)
{
$z_3=\displaystyle
\begin{cases}
z_2 & \text{if } z_2 = x_1,\\[4pt]
\perp & \text{else}.
\end{cases}
$
};

\draw[->] (1,0) -- node[above] {$x_1$} (f1);
\draw[->] (f1) -- node[above] {$x_1$} (5,0);

\draw[->] (7,0) -- node[above] {$z_2$} (f2);
\draw[->] (f2) -- node[above] {$z_3$} (11.5,0);

\draw[dashed,->] (f1.south) to[bend right=60] node[below] {$x_1$} (f2.south);

\end{tikzpicture}
}
\caption{Ideal three-party digital-signature functionality.}\label{threeparty}
\vspace{-0.9cm}
\end{figure}
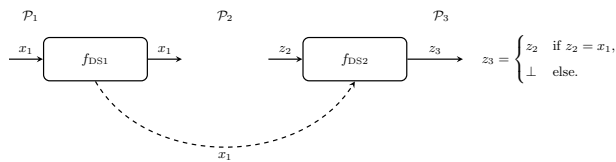
A three-party digital-signature scheme with transfer involves parties $\mathcal{P}_1$ (signer), $\mathcal{P}_2$ (holder), and $\mathcal{P}_3$ (verifier), communicating over pairwise private authenticated channels and an authenticated broadcast primitive \cite{IT-book}. Let $x_1=M$ be the input to the scheme, which is defined by five polynomial-time algorithms:
\begin{itemize}

  \item \textbf{$\mathsf{KeyGen}(1^k)$:}
        run by $\mathcal{P}_1$, outputs $(\mathsf{pk},\mathsf{sk})$.

  \item \textbf{$\mathsf{Sign}_{\mathsf{sk}}(M)$:}
        run by $\mathcal{P}_1$, outputs a signature $\sigma$ to $\mathcal{P}_2$.

  \item \textbf{$\mathsf{Extract}(M,\sigma)$:}
        run by $\mathcal{P}_2$ in the signing phase, outputs
        \(
          z_2 \in \mathcal{X} \cup \{\bot\},
        \)         where $z_2$ is the authenticated value associated with $(M,\sigma)$.

  \item \textbf{$\mathsf{Transfer}(z_2,\sigma)$:}
        run by $\mathcal{P}_2$, outputs $(M',\sigma')$ to $\mathcal{P}_3$.

  \item \textbf{$\mathsf{ExtractTransfer}(M',\sigma')$:}
        run by $\mathcal{P}_3$, outputs
        \(
          z_3 \in \mathcal{X} \cup \{\bot\},
        \)
        representing the verifier's recovered authenticated value.
\end{itemize}
We formalize security using the three-party ideal functionality
$f_{\mathsf{DS}} = (f_{\mathsf{DS1}}, f_{\mathsf{DS2}})$, originally
introduced by Fitzi~\emph{et~al.}
\cite{Fitzi2001,Fitzi2002,FitziWolf2004,IT-book} and illustrated in
\figurename~\ref{threeparty}.
The functionality captures the minimal guarantees required from an
IT-DS scheme with transfer.

\noindent\textbf{Signing Phase $(f_{\mathsf{DS1}})$:}
The signer $\mathcal{P}_1$ provides an input
$x_1 \in \mathcal{X} \cup \{\bot\}$ to the functionality. The value
$x_1$ represents the authenticated value associated with the message. The functionality forwards
$x_1$ to the holder $\mathcal{P}_2$ and stores it internally for use in
the transfer phase. No other party receives information at this stage.

\noindent\textbf{Transfer Phase $(f_{\mathsf{DS2}})$:}
To transfer a signature,  $\mathcal{P}_2$ submits a value
$z_2 \in \mathcal{X}$ to the functionality and $\mathcal{P}_3$
receives
\(
  z_3 = f_{\mathsf{DS2}}(x_1, z_2) =
  \begin{cases}
    z_2, & \text{if } z_2 = x_1,\\[4pt]
    \perp, & \text{otherwise}.
  \end{cases}
\)
The ideal functionality enforces the three core guarantees of
IT-DS:
\begin{itemize}
  \item \textbf{Correctness:} honest executions always yield the
        authenticated value $x_1$.
  \item \textbf{Unforgeability:} neither the holder nor the verifier can
        cause the functionality to output any value other than $x_1$.
  \item \textbf{Transferability:} any honest verifier must accept a value
        previously authenticated by an honest signer and forwarded by an
        honest holder.
\end{itemize}
\subsection{Real-World Security Definition and Supporting Notions}\label{subsec:realworld-security}
We formalize security using the standard real or ideal paradigm. Let
$f_{\mathsf{DS}}$ denote the three-party ideal functionality for digital
signatures introduced by Fitzi~\emph{et al.} \cite{IT-book}. A real protocol $\Pi$ is said to
$\varepsilon$-securely implement $f_{\mathsf{DS}}$ if every real-world attack
can be simulated in the ideal world with at most $\varepsilon$ statistical
distinguishability.
For any adversary $\mathcal{A}$ corrupting at most one party, there must exist a
simulator $\mathcal{S}$ such that the adversary's real-world view when
interacting with $\Pi$ is statistically close to its ideal-world view when
interacting with $\mathcal{S}$ and the ideal functionality $f_{\mathsf{DS}}$, i.e.,
\(
d_{\mathrm{var}}\Bigl(
\mathsf{View}^{\mathrm{real}}_{\mathcal{A}}(\Pi),\,
\mathsf{View}^{\mathrm{ideal}}_{\mathcal{A}}(\mathcal{S}\circ f_{\mathsf{DS}})
\Bigr)
\le \varepsilon.
\)
This definition captures the strongest form of information-theoretic security:
the adversary must not be able to distinguish whether it is interacting with
the real protocol or with the ideal functionality, except with probability
$\varepsilon$.

\noindent\textbf{Characterization via Local Properties:}
Fitzi~\emph{et al.}~\cite{IT-book} established a fundamental equivalence that reduces
simulation-based security to four concrete properties of a protocol.

\begin{theorem}[\cite{IT-book}, Lemma 18.15]\label{main_thm}
A three-party protocol $\Pi$ is a $2\varepsilon$-secure implementation of $f_{\mathsf{DS}}$ with error $\varepsilon$ against one active adversary if the following conditions hold for any attack by the dishonest party:
\begin{itemize}
\item $\varepsilon_{\mathrm{corr}}\mathrm{-correctness}$: with honest $\mathcal{P}_1$ and $\mathcal{P}_2$, $\mathrm{Pr}(z_2=x_1)\geq 1-\varepsilon_{\mathrm{corr}},$
\item $\varepsilon_{\mathrm{uf}}\mathrm{-unforgeability}$: with honest $\mathcal{P}_1$ and $\mathcal{P}_3$, $\mathrm{Pr}(z_3\neq x_1,z_3\neq \bot)\leq \varepsilon_{\mathrm{uf}},$
\item $\varepsilon_{\mathrm{trans}}\mathrm{-transferability}$: with honest $\mathcal{P}_2$ and $\mathcal{P}_3$, $\mathrm{Pr}(z_2\neq z_3,z_2\neq \bot)\leq \varepsilon_{\mathrm{trans}},$
\item $\varepsilon_{\mathrm{sec}}\mathrm{-secrecy}$: with honest $\mathcal{P}_1$ and $\mathcal{P}_2$, $\mathsf{View}_3^\Pi(x_1)$ of $\mathcal{P}_3$ in the signing phase is almost independent of $\mathcal{P}_1$'s input $x_1$, which implies the existence of $Q_{\mathsf{View}_3^\Pi}$ such that $d_{\mathrm{var}}(P_{\mathsf{View}_3^\Pi(x_1)},Q_{\mathsf{View}_3^\Pi})\leq\varepsilon_{\mathrm{sec}}$, for every $x_1$,
\end{itemize}
where
\(
\varepsilon = \max\{
\varepsilon_{\mathrm{corr}},
\varepsilon_{\mathrm{uf}},
\varepsilon_{\mathrm{trans}},
\varepsilon_{\mathrm{sec}}
\}.
\)
\end{theorem}
Thus, to prove that a protocol securely realizes $f_{\mathsf{DS}}$, it suffices
to establish these four properties. This modular approach greatly simplifies the
security analysis in Section~\ref{Security_sec2}.
\section{Quantum Generalization of Theorem~\ref{main_thm}}\label{sec:quantum-lemma-18-15}
We analyze our constructions under three models for the hash function $H$: the \emph{Pure IT model}, where $H:\mathcal{X}\to\mathbb{F}_p$ is a deterministic encoding; the \emph{IT+ROM}, where $H:\{0,1\}^\ast\to\mathbb{F}_p$ is a classical random oracle; and the \emph{IT+QROM}, where $H$ is accessible to quantum adversaries. In this section we lift Theorem~\ref{main_thm} (cf.\ \cite[Lemma~18.15]{IT-book}) to the quantum setting. This extension is necessary because many classical ROM techniques fail in the QROM—e.g., classical rewinding, inspecting superposition queries, classical oracle programming, and the forking lemma \cite{Unruh2010,BonehZhandry2013,Unruh2015,Zhandry2012}. In contrast, the proof of Theorem~\ref{main_thm} never interacts with the adversary algorithmically; it reasons only about the joint distributions of $(X_1,Z_2,Z_3,\mathsf{View}_1,\mathsf{View}_2,\mathsf{View}_3)$, where $X_1$ is the signer's input, $Z_2,Z_3$ are outputs, and each $\mathsf{View}_i$ is the corresponding party's view. The simulators are defined at the level of distributions (or, in the quantum case, density operators), not as black-box procedures that manipulate a quantum adversary. The proof uses only standard IT tools—triangle inequality, conditioning/averaging, and coupling (including maximal couplings)—all of which lift directly to the quantum setting by replacing variational distance with trace distance and interpreting views as quantum states. No rewinding, forking, or oracle programming is required \cite{Watrous2018,Unruh2015}.
\subsection{Quantum Preliminaries}
We briefly recall the quantum-information concepts required for our analysis.
All Hilbert spaces are finite-dimensional. For a Hilbert space $\mathcal{H}$, we
write $\mathcal{D}(\mathcal{H})$ for the set of density operators on
$\mathcal{H}$, i.e., positive semidefinite operators of trace~$1$.

\noindent\textit{Trace of an Operator and Partial Trace:}
Let $T$ be a linear operator acting on a finite-dimensional Hilbert space
$\mathcal{H}$. The \emph{trace} of $T$ is defined as $\mathrm{Tr}[T] := \sum_{i} \langle i \,|\, T \,|\, i \rangle,$
where $\{\,|i\rangle\,\}$ is any orthonormal basis of $\mathcal{H}$.
This quantity is basis independent and equals the sum of the eigenvalues of $T$.
Now, let $\mathcal{H}_A$ and $\mathcal{H}_B$ be Hilbert spaces describing subsystems
$A$ and $B$, and let $\rho_{AB}$ be a density operator on the composite system
$\mathcal{H}_A \otimes \mathcal{H}_B$.
The \emph{partial trace over $B$}, denoted $\mathrm{Tr}_B$, is
\begin{eqnarray}
  \mathrm{Tr}_B[\rho_{AB}]
  :=   \sum_{j}
  \bigl( I_A \otimes \langle j|_B \bigr)\,
  \rho_{AB}\,
  \bigl( I_A \otimes |j\rangle_B \bigr)\equiv \rho_A,
\end{eqnarray}
where $\{|j\rangle_B\}$ is any orthonormal basis of $\mathcal{H}_B$.
The resulting operator $\rho_A := \mathrm{Tr}_B[\rho_{AB}]$ is the reduced state
of subsystem $A$.
The partial trace over $A$ is
\begin{eqnarray}
 \mathrm{Tr}_A[\rho_{AB}]  := \sum_{i}
  \bigl( \langle i|_A \otimes I_B \bigr)\,
  \rho_{AB}\,
  \bigl( |i\rangle_A \otimes I_B \bigr)\equiv \rho_B.
\end{eqnarray}

\noindent\textit{Trace Distance:}
For density operators $\rho,\sigma\in\mathcal{D}(\mathcal{H})$, the
\emph{trace distance} is $\Delta(\rho,\sigma)
  :=  \tfrac12\|\rho-\sigma\|_1,$
where $\|A\|_1=\mathrm{Tr}\sqrt{A^\dagger A}$ denotes the trace norm. Trace
distance is the quantum analogue of variational distance and satisfies the
triangle inequality.
\begin{definition}[Bipartite States, \cite{NielsenChuang2010}]
A bipartite quantum state over registers $A$ and $B$ is a density operator
$\rho_{AB}=\rho_A\otimes \rho_B\in\mathcal{D}(\mathcal{H}_A\otimes\mathcal{H}_B)$.
\end{definition}

\noindent\textit{Classical--Quantum States:}
A classical--quantum (cq) state over a classical register $X$ and a quantum
register $Q$ has the form
\(
  \rho_{XQ}
  =
  \sum_x p_x\,|x\rangle\!\langle x|\otimes \rho_Q^x,
\)
where $\{p_x\}$ is a probability distribution and $\rho_Q^x$ are density
operators. The trace distance between two cq-states decomposes as follows.
\begin{lemma}[Trace-Distance Decomposition, \cite{Watrous2018}]
\label{lem:trace-decomposition}
Let
\(
  \rho_{XQ}=\sum_x p_x\,|x\rangle\!\langle x|\otimes\rho_Q^x\) and
\(
  \sigma_{XQ}=\sum_x p_x\,|x\rangle\!\langle x|\otimes\sigma_Q^x,
\)
be classical--quantum states with the \emph{same} classical marginal
distribution $\{p_x\}$. Then
\(
  \Delta(\rho_{XQ},\sigma_{XQ})
  \le
  \sum_x p_x\,\Delta(\rho_Q^x,\sigma_Q^x).
\)
\end{lemma}

\noindent\textit{CPTP Maps:}
A linear map $\Phi:\mathcal{D}(\mathcal{H})\to\mathcal{D}(\mathcal{H}')$ is
\emph{completely positive and trace-preserving} (CPTP) if
$\Phi\otimes I_{\mathcal{K}}$ is positive for every auxiliary space
$\mathcal{K}$ and $\mathrm{Tr}(\Phi(\rho))=1$ for all $\rho$ \cite{cao2023quantummapscptphptp}. CPTP maps model
all physically admissible quantum operations, including unitary evolution,
measurement, discarding registers, and interaction with an environment.
\begin{lemma}[Contractivity of CPTP Maps, \cite{Watrous2018}] \label{lem:contractivity-prelim}
For any CPTP map $\Phi$ and any density operators $\rho,\sigma$, we have
\(
  \Delta\bigl(\Phi(\rho),\Phi(\sigma)\bigr)
  \;\le\;
  \Delta(\rho,\sigma).
\)
\end{lemma}

\noindent\textit{Quantum Coupling:}
The following lemma is the quantum analogue of the classical coupling argument
used to relate the probability of ``bad'' events in two distributions.
\begin{lemma}[Quantum Coupling Lemma, \cite{Winter2016}]
\label{lem:quantum-coupling-prelim}
Let $\rho_{AB}$ and $\sigma_{AB}$ be bipartite states. If an event $E$ has
probability $0$ under $\sigma_{AB}$, then
\(
  \Pr_{\rho}[E]
  \;\le\;
  \Delta(\rho_{AB},\sigma_{AB}).
\)
\end{lemma}
These tools suffice to lift the classical characterization of digital-signature security to the quantum setting. Every step in the quantum generalization of Theorem~\ref{main_thm} consists only of applying CPTP maps (e.g., tracing out registers, simulating views), conditioning on classical events, and using the triangle inequality for trace distance. Since CPTP maps are contractive, all classical arguments that bound variational distance extend directly by replacing it with trace distance and interpreting views as quantum states. No rewinding, forking, or oracle programming is used, so none of the classical ROM techniques that fail in the QROM arise here.
\subsection{Quantum Security of $f_{\mathsf{DS}}$}
We consider a three-party digital signature functionality
$f_{\mathsf{DS}} = (f_{\mathsf{DS},1}, f_{\mathsf{DS},2})$, $\mathcal{P}_1$, $\mathcal{P}_2$  and $\mathcal{P}_3$ as before.
A protocol $\Pi$ implementing $f_{\mathsf{DS}}$ is an interactive protocol
among $\mathcal{P}_1, \mathcal{P}_2, \mathcal{P}_3$ (and possibly an environment and a quantum adversary)
with classical inputs and outputs but potentially quantum internal computation
and access to a quantum random oracle.
In our setting, each party $\mathcal{P}_i$ has a classical output $Z_i$ and a
quantum view $\mathsf{View}_i$ consisting of all quantum registers and
classical transcripts observed during the protocol. The joint state of all
registers in the real execution is denoted by
$\rho_{x_1 Z_2 Z_3 \mathsf{View}_1 \mathsf{View}_2 \mathsf{View}_3}$, with
lowercase symbols (e.g., $x_1$) representing fixed inputs and uppercase symbols
(e.g., $Z_2,Z_3$) representing random variables produced during the execution.
All CPTP maps are
contractive with respect to trace distance, and the triangle inequality holds.
We adopt a standard simulation-based definition in the quantum setting
(see, e.g.,~\cite{cao2023quantummapscptphptp}).
\begin{definition}[$\varepsilon$-Security in Quantum]
A protocol $\Pi$ is
$\varepsilon$-securely implementing $f_{\mathsf{DS}}$ against one active
quantum adversary if, for each corrupted party, there exists a quantum
polynomial-time simulator such that the real and ideal executions are
$\varepsilon$-close in trace distance (as joint states of the environment and
honest parties).
\end{definition}
We now restate the four properties in the quantum setting. All probabilities
are taken over the randomness of the protocol, the adversary, and any oracles.
\begin{description}
  \item[Correctness:] With honest $\mathcal{P}_1$ and $\mathcal{P}_2$,
  \(
    \Pr[Z_2\neq x_1]\le\varepsilon
  \), for all $x_1\in\mathcal{X}$.
  \item[Unforgeability:] With honest $\mathcal{P}_1$ and $\mathcal{P}_3$,
  \(
    \Pr[Z_3\notin\{x_1,\bot\}]\le\varepsilon,
  \)   for all $x_1\in\mathcal{X}$.

  \item[Transferability:] If $\mathcal{P}_2$ and $\mathcal{P}_3$ are honest,
  then
  \(
    \Pr[Z_2\neq Z_3\;\wedge\;Z_2\neq\bot]\le\varepsilon.
  \)
  \item[Secrecy:] If $\mathcal{P}_1$ and $\mathcal{P}_2$ are honest, then the
  view of $\mathcal{P}_3$ in the signing phase is almost independent of $x_1$ and
  there exists a state $\sigma_{\mathsf{View}_3}$ such that for all $x_1$,
  \(
    \Delta\bigl(\rho_{\mathsf{View}_3(x_1)},\sigma_{\mathsf{View}_3}\bigr)
    \le\varepsilon.
  \)
\end{description}
\begin{theorem}[Quantum Security of Digital Signatures]\label{thm:quantum-lemma-18-15}
Let $\Pi$ be a three-party protocol for $f_{\mathsf{DS}}$ against quantum adversaries.
If $\Pi$ satisfies correctness, unforgeability, transferability, and secrecy with error at most $\varepsilon$, then it is an $O(\varepsilon)$-secure implementation of $f_{\mathsf{DS}}$; conversely, any $\varepsilon$-secure implementation of $f_{\mathsf{DS}}$ satisfies these four properties with error $O(\varepsilon)$.
The hidden constants arise from a constant number of triangle-inequality applications.
\end{theorem}
\begin{proof}
$(\Leftarrow)$ If the four properties hold, then $\Pi$ is secure.

\emph{Case 1: $\mathcal{P}_1$ and $\mathcal{P}_2$ honest.}
When $\mathcal{P}_1$ and  $\mathcal{P}_2$ behave honestly, the
only potentially corrupted party is $\mathcal{P}_3$.  The simulator for
$\mathcal{P}_3$ must therefore reproduce the joint distribution of
$(X_1, Z_2, \mathsf{View}_3)$ as seen in the real protocol.
By secrecy, the adversary controlling $\mathcal{P}_3$ learns essentially nothing
about the signer's input $x_1$.  Formally, there exists a \emph{fixed} state
$\sigma_{\mathsf{View}_3}$ (independent of~$x_1$) such that
\(
  \Delta\!\left(\rho_{\mathsf{View}_3(x_1)},\,\sigma_{\mathsf{View}_3}\right)
  \le \varepsilon .
\)
Thus, in the ideal world the simulator may simply output $\sigma_{\mathsf{View}_3}$
as the view of~$\mathcal{P}_3$.
Correctness ensures that the honest verifier $\mathcal{P}_2$ outputs the signer's
input except with small probability:
\(
  \Pr[\,Z_2 \neq x_1\,] \le \varepsilon .
\)
Hence, the real joint state of $(x_1, Z_2)$ is $\varepsilon$-close to the ideal
joint state in which $Z_2$ is deterministically equal to $x_1$.
We now compare the full real state
$\rho_{x_1 Z_2 \mathsf{View}_3}$ with the ideal state
$\sigma_{x_1 Z_2 \mathsf{View}_3}$, where the ideal state is defined by:
\(
  Z_2 = x_1\), and   \(\mathsf{View}_3 \sim \sigma_{\mathsf{View}_3}.
\)
Applying Lemma~\ref{lem:trace-decomposition} to decompose the trace distance over
the classical registers $(x_1, Z_2)$, and then using the triangle inequality, we
obtain
\[
  \Delta\!\left(
    \rho_{x_1 Z_2 \mathsf{View}_3},
    \sigma_{x_1 Z_2 \mathsf{View}_3}
  \right)
  \le
  \underbrace{\Pr[Z_2 \neq x_1]}_{\le \varepsilon}
  \;+\;
  \underbrace{
    \Delta\!\left(
      \rho_{\mathsf{View}_3(x_1)},
      \sigma_{\mathsf{View}_3}
    \right)
  }_{\le \varepsilon}
  \;\le\; 2\varepsilon .
\]
Thus, the simulator for $\mathcal{P}_3$ can reproduce the ideal joint state up to
trace distance at most $2\varepsilon$, and therefore achieves simulation error
$O(\varepsilon)$ in this case.

\emph{Case 2: $\mathcal{P}_1$ and $\mathcal{P}_3$ honest.}
Here $\mathcal{P}_1$ and  $\mathcal{P}_3$ behave
honestly, while $\mathcal{P}_2$ may be corrupted.  In the ideal world, the
functionality guarantees that the output of $\mathcal{P}_3$ is always either the
correct message $x_1$ or the distinguished symbol $\bot$ indicating rejection.
Thus, the ideal joint distribution of $(x_1,Z_3)$ is supported on the set
$\{(x_1,x_1),(x_1,\bot)\}$.
In the real protocol, an honest $\mathcal{P}_3$ should never accept an incorrect
message.  Unforgeability therefore ensures that
\(
  \Pr[\,Z_3 \notin \{x_1,\bot\}\,] \le \varepsilon .
\)
Equivalently, the real distribution of $(x_1,Z_3)$ is $\varepsilon$-close to the
ideal distribution in which $Z_3$ is always either $x_1$ or $\bot$.
Lemma~\ref{lem:quantum-coupling-prelim} states that whenever two classical--quantum
states differ only in a classical register with error probability at most
$\varepsilon$, there exists a joint coupling whose trace distance is at most
$\varepsilon$.  Applying this lemma to the pair of states
$\rho_{x_1 Z_3}$ (real) and $\sigma_{x_1 Z_3}$ (ideal), we obtain
\(
  \Delta\!\left(\rho_{x_1 Z_3},\,\sigma_{x_1 Z_3}\right) \le \varepsilon .
\)
The corrupted party $\mathcal{P}_2$ obtains its view by applying some
CPTP map to the registers $(x_1,Z_3)$.
Formally, there exists a CPTP map $\Phi$ such that
\[
  \rho_{x_1 Z_3 \mathsf{View}_2}
  = \Phi(\rho_{x_1 Z_3}),
  \qquad\text{and}\qquad
  \sigma_{x_1 Z_3 \mathsf{View}_2}
  = \Phi(\sigma_{x_1 Z_3}) .
\]
By Lemma~\ref{lem:contractivity-prelim}, CPTP maps cannot increase trace
distance.  Therefore,
\[
  \Delta\!\left(
    \rho_{x_1 Z_3 \mathsf{View}_2},
    \sigma_{x_1 Z_3 \mathsf{View}_2}
  \right)
  \le
  \Delta\!\left(
    \rho_{x_1 Z_3},
    \sigma_{x_1 Z_3}
  \right)
  \le \varepsilon .
\]
Thus, the simulator for the corrupted party $\mathcal{P}_2$ can reproduce the
ideal joint state of $(x_1,Z_3,\mathsf{View}_2)$ up to trace distance at most
$\varepsilon$, achieving simulation error $O(\varepsilon)$ in this case.

\emph{Case 3: $\mathcal{P}_2$ and $\mathcal{P}_3$ honest.}
Here  $\mathcal{P}_2$ and  $\mathcal{P}_3$ behave
honestly, while  $\mathcal{P}_1$ may be corrupted.  In the ideal
functionality, whenever the verifier outputs a valid message $Z_2 \neq \bot$, the
transfer recipient must output the same message.  Thus, in the ideal world the
joint distribution of $(Z_2,Z_3)$ satisfies
\(
  Z_2 \neq \bot \;\Longrightarrow\; Z_3 = Z_2 .
\)
In the real protocol, transferability guarantees that an honest $\mathcal{P}_3$
cannot be convinced to output a message different from the one accepted by the
honest $\mathcal{P}_2$.  Formally,
\(
  \Pr[\,Z_2 \neq Z_3 \;\wedge\; Z_2 \neq \bot\,] \le \varepsilon .
\)
Equivalently, the real joint distribution of $(Z_2,Z_3)$ is $\varepsilon$-close
to the ideal distribution in which $Z_3$ always equals $Z_2$ whenever
$Z_2 \neq \bot$.
Lemma~\ref{lem:quantum-coupling-prelim} states that if two classical--quantum
states differ only in a classical register with error probability at most
$\varepsilon$, then there exists a coupling whose trace distance is at most
$\varepsilon$.  Applying this lemma to the real and ideal states of the pair
$(Z_2,Z_3)$ yields
\(
  \Delta\!\left(
    \rho_{Z_2 Z_3},
    \sigma_{Z_2 Z_3}
  \right)
  \le \varepsilon .
\)
The corrupted party $\mathcal{P}_1$ obtains its view by applying some
CPTP map to the registers $(x_1,Z_2,Z_3)$.
Thus, there exists a CPTP map $\Phi$ such that
\[
  \rho_{x_1 Z_2 Z_3 \mathsf{View}_1}
  = \Phi(\rho_{x_1 Z_2 Z_3})
  \qquad\text{and}\qquad
  \sigma_{x_1 Z_2 Z_3 \mathsf{View}_1}
  = \Phi(\sigma_{x_1 Z_2 Z_3}) .
\]
By Lemma~\ref{lem:contractivity-prelim}, CPTP maps cannot increase trace
distance.  Therefore,
\[
  \Delta\!\left(
    \rho_{x_1 Z_2 Z_3 \mathsf{View}_1},
    \sigma_{x_1 Z_2 Z_3 \mathsf{View}_1}
  \right)
  \le
  \Delta\!\left(
    \rho_{Z_2 Z_3},
    \sigma_{Z_2 Z_3}
  \right)
  \le \varepsilon .
\]
Hence, the simulator for the corrupted signer $\mathcal{P}_1$ can reproduce the
ideal joint state up to trace distance at most $\varepsilon$, achieving
simulation error $O(\varepsilon)$ in this case.

$(\Rightarrow)$ If $\Pi$ is secure, then the four properties hold.

\emph{Correctness:}
Assume that $\mathcal{P}_1$ and $\mathcal{P}_2$ behave honestly.  In the ideal
world, the functionality $f_{\mathsf{DS}}$ always delivers the signer's input to
the honest verifier; that is, the ideal output satisfies $Z_2 = x_1$ with
probability~$1$.  Security of $\Pi$ means that for every environment, the real
execution of the protocol is $\varepsilon$-indistinguishable from the ideal
execution.  In particular, the joint real and ideal distributions of
$(x_1,Z_2)$ are within trace distance at most $\varepsilon$.
Since the ideal distribution assigns zero probability to the event
$\{Z_2 \neq x_1\}$, indistinguishability implies that the real execution can
assign at most $\varepsilon$ probability to this event.  Formally,
\(
  \Pr_{\text{real}}[\,Z_2 \neq x_1\,]
  \;\le\;
  \Delta\!\left(
    \rho_{x_1 Z_2}^{\text{real}},
    \rho_{x_1 Z_2}^{\text{ideal}}
  \right)
  \;\le\; \varepsilon .
\)
Thus, an honest verifier outputs the correct message except with probability at
most~$\varepsilon$, establishing correctness.

\emph{Unforgeability:}
Assume that $\mathcal{P}_1$ and $\mathcal{P}_3$ behave honestly.  In the ideal world, the functionality
$f_{\mathsf{DS}}$ guarantees that an honest $\mathcal{P}_3$ can never be made to
accept an incorrect message: its output is always either the true message $x_1$
or the rejection symbol $\bot$.  Thus, the ideal distribution of $Z_3$ is
supported entirely on the set $\{x_1,\bot\}$.
Security of $\Pi$ means that the real execution must be
$\varepsilon$-indistinguishable from the ideal execution.  In particular, the
real and ideal distributions of the classical register $Z_3$ must be within
trace distance at most $\varepsilon$.  Since the ideal distribution assigns zero
probability to the event $\{Z_3 \notin \{x_1,\bot\}\}$, any deviation from this
event in the real world must be bounded by the distinguishing advantage.
Formally, by Lemma~\ref{lem:quantum-coupling-prelim}, whenever two
classical--quantum states differ only on a classical register with error
probability at most $\varepsilon$, there exists a coupling whose trace distance
is at most $\varepsilon$.  Applying this lemma to the real and ideal states of
$Z_3$ yields
\(
  \Pr_{\text{real}}\!\left[\,Z_3 \notin \{x_1,\bot\}\,\right]
  \;\le\;
  \Delta\!\left(
    \rho_{Z_3}^{\text{real}},
    \rho_{Z_3}^{\text{ideal}}
  \right)
  \;\le\; \varepsilon .
\)
Thus, an honest $\mathcal{P}_3$ outputs an invalid message with probability at
most $\varepsilon$, establishing unforgeability.

\emph{Transferability:}
Assume that $\mathcal{P}_2$  and $\mathcal{P}_3$  behave honestly.  In the ideal world, the functionality
$f_{\mathsf{DS}}$ enforces perfect transferability and whenever the verifier
accepts a message, i.e., whenever $Z_2 \neq \bot$, the transfer recipient must
output the same message.  Formally, the ideal distribution satisfies
\(
  Z_2 \neq \bot \;\Longrightarrow\; Z_3 = Z_2 .
\)
Security of $\Pi$ means that the real execution must be
$\varepsilon$-indistinguishable from the ideal execution.  Since the ideal
distribution assigns zero probability to the event
\(
  \{\,Z_2 \neq Z_3 \;\wedge\; Z_2 \neq \bot\,\}
\), the real execution can assign at most $\varepsilon$ probability to this event.
Indeed, if the real protocol caused an honest $\mathcal{P}_3$ to output a value
different from the one accepted by the honest $\mathcal{P}_2$ with probability
greater than~$\varepsilon$, then an environment could distinguish the real and
ideal executions with advantage exceeding~$\varepsilon$, contradicting security.
Thus, indistinguishability immediately yields
\(
  \Pr_{\text{real}}\!\left[\,Z_2 \neq Z_3 \;\wedge\; Z_2 \neq \bot\,\right]
  \;\le\; \varepsilon .
\)
This establishes transferability.

\emph{Secrecy.}
Assume that $\mathcal{P}_1$  and $\mathcal{P}_2$
behave honestly.  In this setting, the only potentially corrupted party is
$\mathcal{P}_3$.  In the ideal world, the functionality $f_{\mathsf{DS}}$
reveals no information about the signer's input $x_1$ to~$\mathcal{P}_3$.
Consequently, the ideal view of $\mathcal{P}_3$ is a \emph{fixed} state,
independent of~$x_1$.  Let us denote this state by
$\sigma_{\mathsf{View}_3}$.
Security of $\Pi$ means that for every environment, the real execution is
$\varepsilon$-indistinguishable from the ideal execution.  In particular, the
real and ideal distributions of the adversary's view must be within trace
distance at most~$\varepsilon$.  Since the ideal view does not depend on~$x_1$,
this indistinguishability condition implies that the real view of
$\mathcal{P}_3$—which may in principle depend on~$x_1$—cannot vary with~$x_1$
by more than~$\varepsilon$ in trace distance.
Formally, letting $\rho_{\mathsf{View}_3(x_1)}$ denote the real view of
$\mathcal{P}_3$ when the signer's input is~$x_1$, we obtain
\(
  \Delta\!\left(
    \rho_{\mathsf{View}_3(x_1)},
    \sigma_{\mathsf{View}_3}
  \right)
  \le \varepsilon .
\)
Thus, an honest $\mathcal{P}_3$ learns at most $\varepsilon$ information about
the signer's input, establishing secrecy.
This completes the proof. \hfill$\square$
\end{proof}
\subsection{Implications for QROM Analysis}
Theorem~\ref{thm:quantum-lemma-18-15} shows that Theorem~\ref{main_thm} is robust to the
quantum setting: once we establish correctness, unforgeability, transferability,
and secrecy against quantum adversaries in the QROM (with respect to classical
inputs/outputs and possibly quantum internal computation), we may soundly
conclude simulation-based security of our implementation of $f_{\mathsf{DS}}$.
The known issues that invalidate some classical ROM reductions in the QROM
(e.g., forking-lemma based proofs) do not apply here, since the lemma and its
proof are purely information-theoretic and metric in nature.

When we prove that a scheme is unforgeable against quantum adversaries in the
QROM, $H : \{0,1\}^\ast \to \mathbb{F}_p$
is modeled as a uniformly random function, and the adversary may query $H$ in
superposition \(
\ket{x,y} \mapsto \ket{x, y \oplus H(x)}
\).
\begin{definition}[Forgery event in quantum]
A \emph{forgery} occurs if a (quantum) adversary $\mathcal{A}$ outputs a pair
$(M^\ast,\sigma_{\mathsf{alg}}^\ast)$ such that  $M^\ast$ was never submitted to the signing oracle, and the verification algorithm accepts $(M^\ast,\sigma_{\mathsf{alg}}^\ast)$ as a valid signature.
We denote this event by $\mathsf{Forge}$.
\end{definition}
We can prove that a quantum adversary making $q_H$ queries to a random oracle
$H : \{0,1\}^\ast \to \mathbb{F}_p$ can cause a collision in the oracle outputs
with probability at most $O(q_H^2/p)$. The key technical tool is the
\emph{recording} (\emph{compressed-oracle}) lemma, which says that the
interaction of a quantum algorithm with a random oracle can be represented as a
distribution over classical query--answer transcripts \cite{Zhandry2019}.
\begin{definition}[Collision Event]
Let $H : \{0,1\}^\ast \to \mathbb{F}_p$ be a random function, and let
$\mathcal{A}$ be a (possibly quantum) oracle algorithm. We say that a
\emph{collision} occurs if there exist distinct inputs $x \neq x'$ such that
$H(x) = H(x')$ and both $x$ and $x'$ appear in the (compressed) oracle
transcript associated with $\mathcal{A}$'s execution. We denote this event by
$\mathsf{Coll}$.
\end{definition}
\begin{lemma}[Recording/Compressed-Oracle Lemma \cite{Zhandry2019}]\label{lem:recording}
Let the hash function $H : \{0,1\}^\ast \to \mathbb{F}_p$ be a uniformly random function, and let
$\mathcal{A}$ be a quantum oracle algorithm making at most $q_H$ queries to
$H$. Then, there exists a classical random variable $T$ (the \emph{transcript})
such that $T$ is a finite sequence of pairs
        \(
          T = \bigl((x_1, H(x_1)), \dots, (x_t, H(x_t))\bigr)
        \) with $t \le q_H$, where the $x_i$ are (classical) inputs and
        $H(x_i)$ are the corresponding oracle outputs. The distribution of $T$ is identical to the distribution of the
        query--answer transcript in a classical lazy-sampling experiment
        where a classical algorithm makes at most $q_H$ (adaptive) queries
        to a random oracle $H : \{0,1\}^\ast \to \mathbb{F}_p$.
In particular, any event that depends only on the multiset
$\{(x_i, H(x_i))\}_i$ (such as the existence of a collision among the
values $H(x_i)$) has the same probability in the quantum execution as in
this classical transcript model.
\end{lemma}
Using Lemma~\ref{lem:recording}, the quantum birthday bound reduces to a purely
classical probability calculation.
\begin{theorem}[Quantum Birthday Bound \cite{Zhandry2012}]\label{thm:q-birthday}
Let $H : \{0,1\}^\ast \to \mathbb{F}_p$ be a uniformly random function, and let
$\mathcal{A}$ be a quantum oracle algorithm making at most $q_H$ queries to $H$.
Let $\mathsf{Coll}$ be the event that there exist distinct inputs
$x \neq x'$ in the (compressed) transcript such that $H(x) = H(x')$. Then
\(
  \Pr[\mathsf{Coll}]
  \;\le\; \frac{q_H(q_H - 1)}{2p}
  \;\le\; \frac{q_H^2}{2p}.
\)
In particular, $\Pr[\mathsf{Coll}] = O(q_H^2/p)$.
\end{theorem}
\begin{algorithm}[tb]
\caption{\textsc{KeyGen}$(1^k)$}\label{keygen}
\footnotesize
\begin{algorithmic}[1]
\State Choose a prime $p\approx 2^k$ and work over $\mathbb{F}_p^\ast$.
\State Select hash function $H$ and $\mathsf{SSS}_{2,2}$ with public weights $w_0,w_1$.
\State Sample $K \xleftarrow{\$} \mathbb{F}_p^\ast$ and compute $(K_0,K_1)=\mathsf{SSS}_{2,2}(K)$.
\State \textbf{Output:} private key $\mathsf{sk}=K$, public key $(w_0,w_1)$.
\end{algorithmic}
\end{algorithm}

\begin{algorithm}[tb]
\caption{\textsc{Sign}$_{\mathsf{sk}}(M)$}\label{sign}
\footnotesize
\begin{algorithmic}[1]
\Require Secret key $\mathsf{sk} = K$ and ephemeral session key $k_{\mathsf{sig}}$.
\State $n \leftarrow \mathrm{HMAC}_{k_{\mathsf{sig}}}(M)$ and $ r \leftarrow H(M,n)$.
\State Sample $\alpha,\beta,b,d \xleftarrow{\$}\mathbb{F}_p^\ast$.
\State Compute $\varepsilon=\alpha\beta$ and $(\varepsilon_0,\varepsilon_1)=\mathsf{SSS}_{2,2}(\varepsilon)$.
\State Compute $K'=\mathrm{HMAC}_{K}(M)$ and $(K_0',K_1')=\mathsf{SSS}_{2,2}(K')$.
\State Compute the algebraic signature components:
\begin{equation}\label{silmarils_eqs}
  \begin{array}{lll}
\sigma_1=b(K'-r), &
\sigma_2=db^{-1}, &
\sigma_3=K_1'd,\\[2pt]
\sigma_4=d\,\varepsilon^{-1}\varepsilon_1, &
\sigma_5=d(K_0'-r\,\varepsilon^{-1}\varepsilon_0). &
\end{array}
\end{equation}
\State \textbf{Output: } $\sigma_{\mathsf{alg}}=(\sigma_1,\ldots,\sigma_5)$.
\end{algorithmic}
\end{algorithm}

\begin{algorithm}[tb]
\caption{\textsc{Verify}$(M,\sigma_{\mathsf{alg}})$}\label{veri}
\footnotesize
\begin{algorithmic}[1]
\Require Ephemeral session key $k_{\mathsf{sig}}$.
\State Parse $\sigma=(\sigma_1,\sigma_2,\sigma_3,\sigma_4,\sigma_5)$.
\State Compute $n \leftarrow \mathrm{HMAC}_{k_{\mathsf{sig}}}(M)$ and $ r \leftarrow H(M,n)$.
\If{$\sigma_4=0$} \Return 0 \Else
\State Compute
\(
V_0=\sigma_1\sigma_2-\sigma_5,\) and \(
V_1=\sigma_1\sigma_2-\sigma_3+r\,\sigma_4.
\)
\State Reconstruct $V=\mathsf{SSS}_{2,2}^{-1}(V_0,V_1)$.
\If{$V=0$} \Return 1 \Else \,\,\Return 0 \EndIf
\EndIf
\end{algorithmic}
\end{algorithm}

\section{Construction of \textsc{SILMARILS}}\label{Construction_sec}
We now give the full specification of our IT-secure signature scheme. The scheme operates over a prime field $\mathbb{F}_p$ and uses perfect $2$-out-of-$2$ Shamir secret sharing $\mathsf{SSS}_{2,2}$ for all shared values. All randomness comes from a TRNG, and $H$ is either a fixed injective encoding (pure IT model) or a random oracle (IT+ROM).  Verification reconstructs a value that must equal $0$ for valid signatures. Unforgeability follows from the algebraic fact that any deviation from the honest equations introduces a nonzero term that is uniformly distributed in $\mathbb{F}_p^\ast$, since the verification constraints form a linear system satisfied by exactly one tuple—the honest signature—and any modification yields a uniformly random nonzero output.
\subsection{Three-Party \textsc{SILMARILS}}
In this section, we present \textsc{SILMARILS} construction in three-party mode. The goal is to enable a signer $\mathcal{P}_1$ to
authenticate a message $M$ so that a holder $\mathcal{P}_2$ can extract
an authenticated value and subsequently transfer it to a verifier
$\mathcal{P}_3$, while ensuring correctness, unforgeability, and
transferability with error at most $\varepsilon = 1/p$.
Our construction  implements an IC
mechanism based on affine authentication over $\mathbb{F}_p$ \cite[Protocol~18.7]{IT-book}. In our
setting, the IC mechanism is applied not to the message itself but to
the \emph{authenticated value}
\(
  x := H(M,\sigma_{\mathsf{alg}}) \in \mathbb{F}_p,
\)
where $\sigma_{\mathsf{alg}} = (\sigma_1,\dots,\sigma_5)$ is the
algebraic part of \textsc{SILMARILS} signature produced by
Algorithm~\ref{sign}. The IC layer ensures that $\mathcal{P}_2$
and $\mathcal{P}_3$ agree on a unique affine authentication pair
$(x,\sigma)$ satisfying
\(
  \sigma = k_1 x + k_2,
\)
for secret coefficients $(k_1,k_2)$ held by $\mathcal{P}_3$. Any attempt
by a malicious $\mathcal{P}_1$ to cause $\mathcal{P}_2$ and
$\mathcal{P}_3$ to disagree on the authenticated value is detected with
probability at least $1 - 1/p$. Hence,
the resulting three-party \textsc{SILMARILS} protocol consists of two
phases:

\noindent\textbf{Signing phase:} The signer $\mathcal{P}_1$ computes $\sigma_{\mathsf{alg}}$ using
    Algorithm~\ref{sign}, derives
    $x = H(M,\sigma_{\mathsf{alg}})$, and initiates the IC protocol with
    $\mathcal{P}_2$ and $\mathcal{P}_3$ to authenticate $x$. The IC
    protocol consists of a \emph{challenge phase} (run by $\mathcal{P}_2$) and
    a consistency-check phase (run by $\mathcal{P}_1$ and
    $\mathcal{P}_3$). At the
    end of this phase, the holder $\mathcal{P}_2$ obtains an
    authenticated pair $(x,\sigma)$ consistent with the verifier's
    secret line $\sigma = k_1 x + k_2$ except with probability $1/p$.

\noindent\textbf{Transfer phase:}     The holder $\mathcal{P}_2$ sends $(x,\sigma)$ to $\mathcal{P}_3$,
    who checks whether the affine relation holds. If so, it outputs
    $z_3 = x$; otherwise, it outputs $\bot$.
To interpret $x$ as a signature on $M$, the verifier recomputes $\sigma_{\mathsf{alg}}$ from $M$, checks that $H(M,\sigma_{\mathsf{alg}})=x$, and verifies the algebraic predicate $V=\mathsf{SSS}_{2,2}^{-1}(V_0,V_1)=0$ from Algorithm~\ref{veri}; both conditions hold for all honest signatures, so the protocol realizes the ideal three-party functionality with error $1/p$. Algorithms~\ref{keygen}, \ref{sign}, and \ref{alg:3p-sign} give key generation, signing, and three-party extraction. Each algebraic signature $\sigma_{\mathsf{alg}}=(\sigma_1,\ldots,\sigma_5)$ consists of masked encodings of $K'$, the randomness $d$, and Shamir shares of $K'$ and $\varepsilon=\alpha\beta$, ensuring that $V=0$ for honest executions. In the three-party setting, Algorithm~\ref{alg:3p-sign} computes $x=H(M,\sigma_{\mathsf{alg}})$ and initiates the IC protocol with $\mathcal{P}_2$ and $\mathcal{P}_3$, where the challenge and consistency checks (Algorithms~\ref{alg:3p-ic-challenge}–\ref{alg:3p-ic-checks}) authenticate $x$ by producing an affine pair $(x,\sigma)$ satisfying $\sigma=k_1 x+k_2$ for verifier-held coefficients. The holder outputs $z_2=x$, which is consistent with the verifier's line except with probability $1/p$. Any malicious deviation violates an algebraic or affine relation, causing the reconstructed value at the verifier to be uniformly random in $\mathbb{F}_p$, so a forgery is accepted with probability at most $1/p$. Finally, the holder transfers $(x,\sigma)$ to $\mathcal{P}_3$ (Algorithm~\ref{alg:3p-transfer}), and the verifier applies Algorithm~\ref{alg:3p-verifytransfer} to output $z_3\in\{x,\bot\}$ depending on whether the affine check holds.

\begin{algorithm}[tb]
\caption{\textsc{Sign}$^{(3)}_{\mathsf{sk}}(M)$ (run by $\mathcal{P}_1$)}\label{alg:3p-sign}
\footnotesize
\begin{algorithmic}[1]
\State Compute $\sigma_{\mathsf{alg}}$ using Algorithm~\ref{sign}.
\State Compute $x := H(M,\sigma_{\mathsf{alg}})$.
\State Sample $k_1,k_2,x',k_2' \xleftarrow{\$} \mathbb{F}_p$.
\State Compute $\sigma = k_1 x + k_2$ and $\sigma' = k_1 x' + k_2'$.
\State Send $(x,x',\sigma,\sigma')$ to $\mathcal{P}_2$ and $(k_1,k_2,k_2')$ to $\mathcal{P}_3$.
\State Proceed to the IC challenge phase.
\end{algorithmic}
\end{algorithm}

\begin{algorithm}[tb]
\caption{IC Challenge Phase (run by $\mathcal{P}_2$)}\label{alg:3p-ic-challenge}
\footnotesize
\begin{algorithmic}[1]
\State Sample $e \xleftarrow{\$} \mathbb{F}_p$.
\State Compute $x_e := x' + e x$ and $\sigma_e := \sigma' + e \sigma$.
\State Broadcast $(e,x_e,\sigma_e)$.
\end{algorithmic}
\end{algorithm}

\begin{algorithm}[tb]
\caption{IC Consistency Checks (run by $\mathcal{P}_1$ and $\mathcal{P}_3$)}
\label{alg:3p-ic-checks}
\footnotesize
\begin{algorithmic}[1]

\Statex \textbf{Checks by $\mathcal{P}_1$ on $\mathcal{P}_2$'s broadcast:}
\State Verify $x_e = x' + e x$ and $\sigma_e = \sigma' + e \sigma$.
\If{either check fails}
  \State Broadcast ``$\mathcal{P}_2$ corrupt'' and broadcast $(x,\sigma)$.
  \State $\mathcal{P}_2$ sets $(x,\sigma)$ to the broadcasted pair.
  \State $\mathcal{P}_3$ updates $k_2$ so that $\sigma = k_1 x + k_2$.
  \State \textbf{Terminate signing phase with $z_2 = x$.}
\Else
  \State Broadcast ``accept''.
\EndIf

\Statex \textbf{Checks by $\mathcal{P}_3$:}
\State Verify
  \(
    \sigma_e \stackrel{?}{=} k_1 x_e + k_2' + e k_2.
  \)
\State Broadcast ``accept'' or ``reject'' accordingly.

\Statex \textbf{Checks by $\mathcal{P}_1$ on $\mathcal{P}_3$'s broadcast:}
\State Using $(k_1,k_2,k_2')$, verify whether $\mathcal{P}_3$'s broadcast
       (``accept'' or ``reject'') is consistent with the equation above.
\If{$\mathcal{P}_3$'s broadcast is inconsistent}
  \State Broadcast ``$\mathcal{P}_3$ corrupt''.
\EndIf

\If{$\mathcal{P}_3$ says ``accept'' and $\mathcal{P}_1$ says ``accept''}
  \State \textbf{Terminate signing phase with $z_2 = x$.}

\ElsIf{$\mathcal{P}_3$ says ``reject'' and $\mathcal{P}_1$ says ``accept''}
  \State $\mathcal{P}_1$ broadcasts $(x,\sigma)$.
  \State $\mathcal{P}_2$ sets $(x,\sigma)$ to the broadcasted pair.
  \State $\mathcal{P}_3$ updates $k_2$ so that $\sigma = k_1 x + k_2$.
  \State \textbf{Terminate signing phase with $z_2 = x$.}

\ElsIf{$\mathcal{P}_1$ says ``$\mathcal{P}_3$ corrupt''}
  \State $\mathcal{P}_1$ broadcasts $(k_1,k_2)$.
  \State $\mathcal{P}_3$ sets $(k_1,k_2)$ to the broadcasted pair.
  \State $\mathcal{P}_2$ updates $\sigma$ so that $\sigma = k_1 x + k_2$.
  \State \textbf{Terminate signing phase with $z_2 = x$.}
\EndIf
\end{algorithmic}
\end{algorithm}
\begin{algorithm}[tb]
\caption{\textsc{Transfer}$(x,\sigma)$ (run by $\mathcal{P}_2$)}\label{alg:3p-transfer}
\footnotesize
\begin{algorithmic}[1]
\State Send $(x,\sigma)$ to $\mathcal{P}_3$.
\end{algorithmic}
\end{algorithm}

\begin{algorithm}[tb]
\caption{\textsc{ExtractTransfer}$(x,\sigma)$ (run by $\mathcal{P}_3$)}\label{alg:3p-verifytransfer}
\footnotesize
\begin{algorithmic}[1]
\If{$\sigma = k_1 x + k_2$}
  \State Output $z_3 := x$.
\Else
  \State Output $z_3 := \bot$.
\EndIf
\end{algorithmic}
\end{algorithm}
Upon receiving $x$, the verifier recomputes $\sigma_{\mathsf{alg}}$ from
$M$ and checks that $H(M,\sigma_{\mathsf{alg}})=x$ and
$V(P,M,\sigma_{\mathsf{alg}})=0$. If both hold, the signature is
accepted.

\subsection{Two-Party \textsc{SILMARILS}}
Let $H : \{0,1\}^\ast \to \mathbb{F}_p$ with $p \approx 2^k$, be a hash function modeled as a
random oracle, and let $\mathrm{HMAC}$ be a pseudorandom function keyed
by the long-term signing key $K$.
The two-party signature scheme
$\mathsf{SILMARILS} = (\mathsf{KeyGen},\mathsf{Sign},\mathsf{Verify})$
between $\mathcal{P}_1$ and $\mathcal{P}_3$ is defined as follows.
We have also provided an open-source reference implementation in Rust~\cite{silmarils-implementation}.
\begin{description}
\item[\textsf{KeyGen}$(1^k)$:]
Choose a prime $p\approx 2^k$ and work over $\mathbb{F}_p^\ast$.
Fix $H$ and $\mathsf{SSS}_{2,2}$ (with fresh internal randomness each use).
Sample public weights $w_0,w_1 \xleftarrow{\$}\mathbb{F}_p^\ast$ and a long-term key $K \xleftarrow{\$}\mathbb{F}_p^\ast$.
Signer and verifier share $k_{\mathsf{sig}}$ for deriving $n=\mathrm{HMAC}_{k_{\mathsf{sig}}}(M)$.
Output $\mathsf{sk}=K$ and $\mathsf{pk}=(w_0,w_1)$.

\item[\textsf{Sign}$_{\mathsf{sk}}(M)$:]
On input $M$ and $\mathsf{sk}=K$, sample $\alpha,\beta,b,d \xleftarrow{\$} \mathbb{F}_p^\ast$ and set $\varepsilon = \alpha\beta$,
$(\varepsilon_0,\varepsilon_1) = \mathsf{SSS}_{2,2}(\varepsilon)$.
Compute the shared nonce
\(
   n := \mathrm{HMAC}_{k_{\mathsf{sig}}}(M),
\)
and compute
\(
   r := H(M,n).
\)
Derive the per-message signing key
\(
   K' := \mathrm{HMAC}_{K}(M),
\)
and compute its shares
\(
   (K_0',K_1') = \mathsf{SSS}_{2,2}(K').
\)
Then compute the algebraic signature $\sigma_{\mathsf{alg}} = (\sigma_1,\sigma_2,\sigma_3,\sigma_4,\sigma_5)$ using \eqref{silmarils_eqs}.
\item[\textsf{Verify}$_{\mathsf{pk}}(M,\sigma_{\mathsf{alg}})$:]
On input $(M,\sigma_{\mathsf{alg}})$ and public key $(w_0,w_1)$, parse $\sigma_{\mathsf{alg}} = (\sigma_1,\sigma_2,\sigma_3,\sigma_4,\sigma_5)$.
Compute
\(
   n := \mathrm{HMAC}_{k_{\mathsf{sig}}}(M),\ r := H(M,n).
\)
If $\sigma_4 = 0$, output $0$ (reject).
Compute
\begin{equation}\label{silmarils_ver}
    V_0 = \sigma_1\sigma_2 - \sigma_5,\qquad V_1 = \sigma_1\sigma_2 - \sigma_3 + r\,\sigma_4.
\end{equation}
Reconstruct
\(
   V = \mathsf{SSS}_{2,2}^{-1}(V_0,V_1).
\)
If $V = 0$ output $1$, otherwise output $0$.
\end{description}
\begin{remark}[Relationship Between Modes]
The two-party TDV and three-party modes are distinct primitives with different assumptions and guarantees; we present them in a common algebraic framework but do not claim any reduction between them.
\end{remark}
\section{Designated-Verifier Security of Two-Party \textsc{SILMARILS}}\label{sec:dv-security}
In this section we formalize the DV security of the
two-party \textsc{SILMARILS} mode and prove that it satisfies the Jakobsson--Sako--%
Impagliazzo (JSI) notion of DV proofs~\cite{JSI96}.
Intuitively, a designated-verifier signature (DVS) scheme is one in which the
DV can always simulate accepting transcripts that are
indistinguishable from real ones, so that no third party can be convinced
that the signer actually produced a given transcript.

Jakobsson \emph{et al.} motivate designated verification as reconciling authenticity with privacy: ``we designate a verifier when we ensure that nobody but this participant can be convinced by the proof.'' Their core idea is that Alice proves the disjunction ''either $\varphi$ is true, or I am Bob,'' which convinces Bob but is useless to any third party, since Bob can always simulate being Bob. This captures the requirement that any transcript convincing to the designated verifier must be simulatable by that verifier alone, and therefore non-transferable.

We adopt the JSI model with three parties: a prover $\mathcal{P}$ (the signer),
a designated verifier $\mathcal{V}$, and a third party $\mathcal{C}$ (Cindy).
Let $\varphi$ be the statement ``$\sigma$ is a valid \textsc{SILMARILS} signature on
message $M$ under the shared long-term key $K$ and per-pair key
$k_{\mathsf{sig}}$''.  A (possibly interactive) protocol
$(P_{\mathcal{P}},P_{\mathcal{V}})$ is a designated-verifier proof of $\varphi$
if the following two properties hold: 1) \textbf{Correctness for the DV:} If $\mathcal{P}$ is honest and $\varphi$ is true, then an honest
        $\mathcal{V}$ accepts with overwhelming probability; 2) \textbf{Simulatability for the DV:} For any (possibly malicious) prover $\mathcal{P}^\ast$, there exists a
        probabilistic polynomial-time (PPT) simulator $S_{\mathcal{V}}$ that, given
        only $\mathcal{V}$'s secret key and the public inputs, outputs
        transcripts that are computationally indistinguishable from those
        produced by an interaction between $\mathcal{P}^\ast$ and an honest
        $\mathcal{V}$.

Following JSI, we say that $\mathcal{V}$ is a \emph{designated verifier} if for
any protocol $(P_{\mathcal{P}},P_{\mathcal{V}},P_{\mathcal{C}})$ in which
$\mathcal{V}$ attempts to convince $\mathcal{C}$ of $\varphi$, there exists a
protocol $(P'_{\mathcal{V}},P_{\mathcal{C}})$ such that $\mathcal{V}$ alone can
generate transcripts that $\mathcal{C}$ cannot distinguish from those of the
original three-party protocol.  In particular, a hidden verifier $\mathcal{C}$
cannot obtain transferable conviction, since any transcript she sees could have
been generated entirely by $\mathcal{V}$.
In the sequel, we first analyze the algebraic structure of the verification equations, showing that any party who can compute the per-message value $r=H(M,\mathrm{HMAC}_{k_{\mathsf{sig}}}(M))$ can generate accepting transcripts and extract the hidden parameters $a$ and $s=K'$. We then show that this capability is \emph{exactly} the JSI designated-verifier simulation property and therefore a \emph{security feature}, not a flaw.
\subsection{Algebraic Structure and Simulator Construction}
Using Shamir shares $K_0'=K'+a_K w_0$, $K_1'=K'+a_K w_1$, $\varepsilon_0=\varepsilon+a_\varepsilon w_0$, $\varepsilon_1=\varepsilon+a_\varepsilon w_1$, the verifier computes
\(
V_0=\sigma_1\sigma_2-\sigma_5=d(-a_K w_0-r+r\varepsilon^{-1}\varepsilon_0)\) and
\(V_1=\sigma_1\sigma_2-\sigma_3+r\sigma_4=d(-a_K w_1-r+r\varepsilon^{-1}\varepsilon_1).
\)
Using $\varepsilon^{-1}\varepsilon_i=1+\tfrac{a_\varepsilon w_i}{\varepsilon}$, we obtain
\(
V_0=d w_0 C,\) \(V_1=d w_1 C,\) and  \(C:=-a_K+r\tfrac{a_\varepsilon}{\varepsilon}.
\)
Thus $(V_0,V_1)$ are evaluations of $f(x)=dCx$, and reconstruction yields $f(0)=0$, guaranteeing acceptance.

\noindent\textbf{Simulator (DV Forgery):}
Any party knowing $r$ can generate a valid transcript. Compute
$n=\mathrm{HMAC}_{k_{\mathsf{sig}}}(M)$ and $r=H(M,n)$, choose arbitrary
$K'^\ast,a_K,a_\varepsilon\in\mathbb{F}_p$ and $d,\varepsilon,b\in\mathbb{F}_p^\ast$, and form
$K_0'^\ast=K'^\ast+a_K w_0$, $K_1'^\ast=K'^\ast+a_K w_1$,
$\varepsilon_0^\ast=\varepsilon+a_\varepsilon w_0$,
$\varepsilon_1^\ast=\varepsilon+a_\varepsilon w_1$.
Substitute these values into the signature equations~\eqref{silmarils_eqs} and output the resulting $\sigma^\ast=(\sigma_1^\ast,\ldots,\sigma_5^\ast)$.
\subsection{Extraction of the Hidden Parameters}
From the signature equations
\(
\sigma_1=b(s-r),\ \sigma_2=db^{-1},\ \sigma_3=d(a w_1+s),\
\sigma_4=d u_1,\ \sigma_5=d(a w_0+s-r u_0),
\)
with $u_i=\varepsilon^{-1}\varepsilon_i$, the verifier forms the ratio
\(
R=\frac{\sigma_1\sigma_2}{\sigma_3}=\frac{s-r}{a w_1+s},
\)
yielding $s(R-1)+a w_1 R+r=0$. From reconstruction,
\(
a C_w=r C_\varepsilon\), \(C_w=\lambda_0 w_0+\lambda_1 w_1,
\)
with $C_\varepsilon=\lambda_0(u_0-1)+\lambda_1(u_1-1)$ and
\[
d=\sigma_3/(a w_1+s),\quad
u_1=\sigma_4\frac{a w_1+s}{\sigma_3},\quad
u_0=\frac{a w_0+s-\sigma_5\frac{a w_1+s}{\sigma_3}}{r}.
\]
Thus $C_\varepsilon=\alpha_a a+\alpha_s s+\alpha_0$, giving the linear system
\[
(C_w-r\alpha_a)a-r\alpha_s s=r\alpha_0,\qquad
w_1 R\,a+(R-1)s=-r.
\]
Solving yields explicit closed forms for $a$ and $s$. If $R\neq 1$, then $s=\frac{-a w_1 R-r}{R-1}$. Once $s$ is known, $K'_0=a w_0+s$ and $K'_1=a w_1+s$.

\begin{theorem}[DV Security of Two-Party \textsc{SILMARILS}]
The two-party \textsc{SILMARILS} is a DV signature scheme in the sense of JSI notion. In particular, for every message $M$ and every honestly generated transcript $(M,\sigma,r)$, the designated verifier---and only the designated verifier---can (i) verify correctness and (ii) generate transcripts that are computationally indistinguishable from honestly generated ones. Consequently, no third party can obtain transferable conviction that the signer produced the transcript.
\end{theorem}

\begin{proof}
Let $\mathcal{P}_1$ be the signer, $\mathcal{P}_3$ the DV, and $\mathcal{C}$ an arbitrary third party. The signer and DV share the per-pair key $k_{\mathsf{sig}}$, enabling $\mathcal{P}_3$ to compute $n=\mathrm{HMAC}_{k_{\mathsf{sig}}}(M)$ and $r=H(M,n)$ for any $M$. Knowledge of $r$ is the trapdoor enabling simulation, analogous to the trapdoor commitment key in JSI.

\noindent\textit{Correctness:}
For an honestly generated signature $\sigma=(\sigma_1,\dots,\sigma_5)$, the verifier computes
\(
V_0=d w_0 C,\ V_1=d w_1 C,\ C=-a_K+r\tfrac{a_\varepsilon}{\varepsilon},
\)
which interpolate to the affine polynomial $f(x)=dCx$ with $f(0)=0$. Thus the verification algorithm always accepts honest signatures.

\noindent\textit{Simulatability:}
Given any message $M$, the DV can generate an accepting transcript without interacting with the signer. It computes
$r=H(M,\mathrm{HMAC}_{k_{\mathsf{sig}}}(M))$, chooses arbitrary
$K'^\ast,a_K,a_\varepsilon\in\mathbb{F}_p$ and $d,\varepsilon,b\in\mathbb{F}_p^\ast$, and forms
$K_0'^\ast=K'^\ast+a_K w_0$, $K_1'^\ast=K'^\ast+a_K w_1$,
$\varepsilon_0^\ast=\varepsilon+a_\varepsilon w_0$,
$\varepsilon_1^\ast=\varepsilon+a_\varepsilon w_1$.
Substituting into the signature equations~\eqref{silmarils_eqs} yields $\sigma^\ast=(\sigma_1^\ast,\ldots,\sigma_5^\ast)$.
Verification of $\sigma^\ast$ produces the same $(V_0,V_1)$ structure as an honest signature and reconstructs $f(0)=0$, so $\sigma^\ast$ is always accepted.

\noindent\textit{Indistinguishability:}
In honest signing, the parameters $(K',a_K,a_\varepsilon,d,\varepsilon,b)$ are uniformly random in the same domains as in simulation. The only difference is that $K'$ is derived as $\mathrm{HMAC}_K(M)$, whereas $K'^\ast$ is chosen uniformly. Since $K$ is unknown to any external observer, the distributions of honest and simulated transcripts are identical. Hence no algorithm can distinguish them.

\noindent\textit{Non-transferability:}
Let $\mathcal{C}$ receive $(M,\sigma,r)$ from $\mathcal{P}_3$. Because $\mathcal{P}_3$ can generate such transcripts without interacting with $\mathcal{P}_1$, $\mathcal{C}$ cannot determine whether the transcript originated from the signer or from the DV. This matches the JSI definition: for any three-party protocol $(P_{\mathcal{P}},P_{\mathcal{V}},P_{\mathcal{C}})$, there exists a two-party protocol $(P'_{\mathcal{V}},P_{\mathcal{C}})$ producing indistinguishable transcripts.

Hence, only the DV can be convinced by a transcript, and any transcript that convinces them is simulatable by them alone. Therefore, two-party \textsc{SILMARILS} satisfies the JSI designated-verifier definition.
\hfill$\square$
\end{proof}
\section{EUF-CMA Security of Two-Party \textsc{SILMARILS} Against Non-Designated Parties}\label{Security_sec}
In this section we prove that two-party \textsc{SILMARILS} is EUF-CMA secure in the ROM and QROM against non-designated parties.
In the two-party mode of \textsc{SILMARILS}, the designated verifier holds the
shared secret $k_{\mathsf{sig}}$ and can therefore compute the per-message
nonce $n = \mathrm{HMAC}_{k_{\mathsf{sig}}}(M)$ and the receipt
$r = H(M,n)$ for any message $M$. This enables the designated verifier to
simulate accepting transcripts that are indistinguishable from honestly
generated ones, as required by JSI notion. Consequently, standard EUF-CMA
unforgeability cannot hold for the designated verifier. Instead, we prove
unforgeability against all \emph{non-designated} parties, who do not know
$k_{\mathsf{sig}}$ and therefore cannot compute $r$ for new messages.
We emphasize that in the TDV setting the designated verifier may publish a
receipt $r$ enabling third-party verification, but verification is impossible
without $r$, and even with $r$ no external party can determine whether a valid
transcript was produced by the signer or simulated by the designated verifier.
\begin{definition}[EUF-CMA Security for Non-Designated Verifiers]
We write $\mathsf{EUF\text{-}CMA}^{\neg\mathsf{DV}}$ to denote existential
unforgeability under chosen-message attack for all parties except the
designated verifier. An adversary $\mathcal{A}$ is given the public parameters
$(w_0,w_1)$ and oracle access to $\mathsf{Sign}_{\mathsf{sk}}(\cdot)$. It
outputs $(M^\star,\sigma^\star)$. We say that $\mathcal{A}$ wins if
$\mathsf{Verify}_{\mathsf{pk}}(M^\star,\sigma^\star)=1$ and $M^\star$ was never
queried to the signing oracle. The designated verifier is excluded from this
game, as it must be able to simulate signatures.
\end{definition}
We isolate the algebraic core used to analyze unforgeability for non-designated verifiers.
Here the only value depending on $k_{\mathsf{sig}}$ is the receipt $r=H(M,n)$, sampled
uniformly per query; the adversary never learns $n$ or $k_{\mathsf{sig}}$.
\begin{definition}[Algebraic Core Experiment]
In $\mathsf{CoreForge}^{\mathcal{A}}(1^k)$, the challenger samples
$K \xleftarrow{\$}\mathbb{F}_p^\ast$, gives $(w_0,w_1)$ to $\mathcal{A}$, and for each
signing query $M$ samples $\alpha,\beta,b,d \xleftarrow{\$}\mathbb{F}_p^\ast$, sets
$\varepsilon=\alpha\beta$ with $(\varepsilon_0,\varepsilon_1)=\mathsf{SSS}_{2,2}(\varepsilon)$,
samples $r\xleftarrow{\$}\mathbb{F}_p$, computes $K'=\mathrm{HMAC}_K(M)$ with
$(K_0',K_1')=\mathsf{SSS}_{2,2}(K')$, and returns $(\sigma_1,\ldots,\sigma_5)$ as in \eqref{silmarils_eqs}.
On output $(M^\star,\sigma_1^\star,\ldots,\sigma_5^\star)$, the challenger samples
$r^\star\xleftarrow{\$}\mathbb{F}_p$, computes $V_0^\star$ and $V_1^\star$ as in \eqref{silmarils_ver}, and
reconstructs $V^\star=\mathsf{SSS}_{2,2}^{-1}(V_0^\star,V_1^\star)$, and outputs $1$ iff
$M^\star$ was never queried and $V^\star=0$.
The advantage is $\Pr[\mathsf{CoreForge}^{\mathcal{A}}(1^k)=1]$.
\end{definition}
\begin{lemma}[Algebraic Core Unforgeability]\label{algcore}
If $r^\star$ is sampled uniformly independent of the adversary's view \emph{after} the PPT adversary outputs its forgery, then \(
  \Pr[\mathsf{CoreForge}^{\mathcal{A}}(1^k)=1] \le \frac{1}{p},
\)
  which is negligible in $k$.
\end{lemma}
\begin{proof}
Fix any view of $\mathcal{A}$ up to the moment it outputs
$(M^\star,\sigma_1^\star,\ldots,\sigma_5^\star)$. The value $r^\star$ used in
verification is sampled \emph{after} the forgery is output and is uniform in
$\mathbb{F}_p$, independent of $\mathcal{A}$'s view. The verification equations
define $V_0^\star$ independently of $r^\star$ and $V_1^\star$ as an affine
function of $r^\star$. Since $\mathsf{SSS}_{2,2}^{-1}$ is linear, the predicate
$V^\star=0$ holds for at most one value of $r^\star$. Thus the probability that
$V^\star=0$ is at most $1/p$, even for an unbounded adversary.
\hfill$\square$
\end{proof}
\begin{remark}
If the receipt were public, i.e.\ $r=H(M)$ without the hidden nonce $n=\mathrm{HMAC}_{k_{\mathsf{sig}}}(M)$, then the verification equations $V_0=\sigma_1\sigma_2-\sigma_5$ and $V_1=\sigma_1\sigma_2-\sigma_3+r\,\sigma_4$ become public linear constraints on $(\sigma_1,\ldots,\sigma_5)$, making the algebraic core deterministically forgeable: for any new $M'\neq M$, an adversary may choose arbitrary $\sigma_1',\sigma_2',\sigma_3',\sigma_4'\in\mathbb{F}_p$, compute $V_1'=\sigma_1'\sigma_2'-\sigma_3'+r'\sigma_4'$ with $r'=H(M')$, and set $\sigma_5'=\sigma_1'\sigma_2'-(w_0/w_1)V_1'$, which always yields $V'=\mathsf{SSS}_{2,2}^{-1}(V_0',V_1')=0$ and thus a perfect forgery. To prevent this, \textsc{SILMARILS} derives a secret per-message nonce $n=\mathrm{HMAC}_{k_{\mathsf{sig}}}(M)$ and defines $r=H(M,n)$, ensuring $r$ is unpredictable for new messages and eliminating the algebraic forgery.
\end{remark}
\begin{definition}
Let $\Pi_{\mathsf{core}} = (\mathsf{KeyGen},\mathsf{Sign},\mathsf{Verify})$
denote the two-party algebraic core of \textsc{SILMARILS}. We say that
$\Pi_{\mathsf{core}}$ is $\varepsilon_{\mathrm{corr}}$-correct if, for all
messages $M$ and all honestly generated key pairs
$(\mathsf{pk},\mathsf{sk}) \leftarrow \mathsf{KeyGen}$,
\[
  \Pr\bigl[\mathsf{Verify}_{\mathsf{pk}}(M,\sigma_{\mathsf{alg}})=1
  \;\big|\;
  \sigma_{\mathsf{alg}} \leftarrow \mathsf{Sign}_{\mathsf{sk}}(M)
  \bigr]
  \;\ge\; 1 - \varepsilon_{\mathrm{corr}}.
\]
If $\varepsilon_{\mathrm{corr}} = 0$, the algebraic core is said to be
perfectly correct.
\end{definition}

\begin{theorem}[$1/p-$Correctness]
For any message $M$ and every honestly generated key pair $(\mathsf{pk},\mathsf{sk})$, an honestly generated signature $\sigma_{\mathsf{alg}} \leftarrow \mathsf{Sign}_{\mathsf{sk}}(M)$ is rejected with probability at most $1/p$.
\end{theorem}

\begin{proof}
Let $M$ be arbitrary and let
$\sigma_{\mathsf{alg}} = (\sigma_1,\sigma_2,\sigma_3,\sigma_4,\sigma_5)$
be the signature produced by $\mathsf{Sign}_{\mathsf{sk}}(M)$.
Let
\(
  n = \mathrm{HMAC}_{k_{\mathsf{sig}}}(M),
  r = H(M,n),
\)
and let
\(
  K' = \mathrm{HMAC}_{K}(M),
  (K_0',K_1') = \mathsf{SSS}_{2,2}(K').
\)
The signing algorithm samples $\alpha,\beta,b,d \xleftarrow{\$}\mathbb{F}_p^\ast$,
sets $\varepsilon=\alpha\beta$ and
$(\varepsilon_0,\varepsilon_1)=\mathsf{SSS}_{2,2}(\varepsilon)$, and computes $(\sigma_1,\ldots,\sigma_5)$ as in \eqref{silmarils_eqs}.
The verifier computes $V_0,V_1$ as in \eqref{silmarils_ver}
and reconstructs
\(
  V = \mathsf{SSS}_{2,2}^{-1}(V_0,V_1).
\)
We have $V_0 = d\Bigl[(K' - r) - (K_0' - r\,\varepsilon^{-1}\varepsilon_0)\Bigr]$ and $V_1 = d\Bigl[(K' - r) - (K_1' - r\,\varepsilon^{-1}\varepsilon_1)\Bigr]$.
Since $(K_0',K_1')$ are the Shamir shares of $K'$ at points $(w_0,w_1)$,
\(
  K' = \frac{w_0 K_1' - w_1 K_0'}{w_0 - w_1}.
\)
Similarly, since $(\varepsilon_0,\varepsilon_1)$ are the shares of $\varepsilon$,
\(
  \varepsilon = \frac{w_0\varepsilon_1 - w_1\varepsilon_0}{w_0 - w_1}.
\)
Reconstructing $V$ gives
\(
  V = \frac{w_0 V_1 - w_1 V_0}{w_0 - w_1}.
\)
Substituting the expressions for $V_0$ and $V_1$ and simplifying using the
identities above yields
\(
  V = 0.
\)
Considering the case $\sigma_4 = 0$, which causes the verification to reject, the verification predicate evaluates to zero with probability $1 - 1/p$ for honestly generated signatures.
\hfill$\square$
\end{proof}
\subsection{$\mathsf{EUF\text{-}CMA}^{\neg\mathsf{DV}}$ Security of Two-Party \textsc{SILMARILS} in the ROM}

We now prove that the two-party mode of \textsc{SILMARILS} is
$\mathsf{EUF\text{-}CMA}^{\neg\mathsf{DV}}$ secure in the ROM, using
Lemma~\ref{algcore}. Recall that non-designated verifiers do not know the shared
secret $k_{\mathsf{sig}}$ and therefore cannot compute the per-message nonce
$n = \mathrm{HMAC}_{k_{\mathsf{sig}}}(M)$ or the receipt $r = H(M,n)$ for new
messages. This prevents them from mounting the deterministic algebraic forgery
described in the previous remark.

\begin{theorem}
Assume $H$ is modeled as a random oracle and is collision resistant, and
$\mathrm{HMAC}$ is a PRF keyed by $k_{\mathsf{sig}}$. Then the two-party
\textsc{SILMARILS} scheme is $\mathsf{EUF\text{-}CMA}^{\neg\mathsf{DV}}$ secure.
More precisely, for any PPT adversary $\mathcal{A}$ there exist PPT algorithms
$\mathcal{B}_1,\mathcal{B}_2,\mathcal{B}_3$ such that
\[
  \mathsf{Adv}^{\mathsf{EUF\text{-}CMA}^{\neg\mathsf{DV}}}_{\mathsf{SILMARILS}}(\mathcal{A})
  \;\leq\;
  \mathsf{Adv}^{\mathsf{core}}(\mathcal{B}_1)
  + \mathsf{Adv}^{\mathsf{PRF}}_{\mathrm{HMAC}}(\mathcal{B}_2)
  + \mathsf{Adv}^{\mathsf{CR}}_{H}(\mathcal{B}_3)
  + \frac{1}{p}.
\]
\end{theorem}
\begin{proof}
Let $\mathcal{A}$ be a PPT adversary in the
$\mathsf{EUF\text{-}CMA}^{\neg\mathsf{DV}}$ game. We proceed via a sequence of
games.

\noindent\textbf{Game~0:}
This is the real $\mathsf{EUF\text{-}CMA}^{\neg\mathsf{DV}}$ experiment for
\textsc{SILMARILS}. The adversary receives $(w_0,w_1)$, has access to the
signing oracle, and interacts with $H$ as a random oracle. Let $\mathsf{Succ}_0$
denote the event that $\mathcal{A}$ outputs a valid forgery
$(M^\star,\sigma_{\mathsf{alg}}^\star)$ on a message $M^\star$ not previously
signed.

\noindent\textbf{Game~1:}
Replace $\mathrm{HMAC}_{k_{\mathsf{sig}}}$ by a truly random function
$F : \{0,1\}^\ast \to \mathbb{F}_p$. All other aspects remain unchanged.
Let $\mathsf{Succ}_1$ be the forgery event in Game~1. By a standard PRF
reduction, there exists a PPT distinguisher $\mathcal{B}_2$ such that
\(
  \bigl|\Pr[\mathsf{Succ}_0] - \Pr[\mathsf{Succ}_1]\bigr|
  \;\leq\;
  \mathsf{Adv}^{\mathsf{PRF}}_{\mathrm{HMAC}}(\mathcal{B}_2).
\)
From now on we analyze Game~1, where $n := F(M)$ and $r := H(M,n)$.

\noindent\textbf{Game~2:}
We now make explicit how the simulator programs the random oracle $H$.
A simulator $\mathcal{B}$ maintains a table $\mathcal{T}_H$.
On a signing query $M$, it: 1) samples $n \xleftarrow{\$} \mathbb{F}_p$ as $F(M)$; 2) samples $r \xleftarrow{\$} \mathbb{F}_p$; 3) programs $H(M,n) := r$ in $\mathcal{T}_H$; 4) samples $\alpha,\beta,b,d \xleftarrow{\$} \mathbb{F}_p^\ast$,
        sets $\varepsilon=\alpha\beta$ and
        $(\varepsilon_0,\varepsilon_1)=\mathsf{SSS}_{2,2}(\varepsilon)$; 5) computes $K'=\mathrm{HMAC}_{K}(M)$ and
        $(K_0',K_1')=\mathsf{SSS}_{2,2}(K')$; 6) computes $(\sigma_1,\ldots,\sigma_5)$ exactly as in $\mathsf{Sign}$; 7) returns $\sigma_{\mathsf{alg}}$ to $\mathcal{A}$.
For any other $H$-query $(x,y)$ not equal to $(M,n)$ for a signed message,
$\mathcal{B}$ returns a fresh uniform value in $\mathbb{F}_p$, consistent with
$\mathcal{T}_H$.
This perfectly simulates Game~1, so
\(
  \Pr[\mathsf{Succ}_1] = \Pr[\mathsf{Succ}_2].
\)

\noindent\textbf{Analysis of a forgery in Game~2:}
Suppose $\mathcal{A}$ outputs a forgery
$(M^\star,\sigma_{\mathsf{alg}}^\star)$ with
$\sigma_{\mathsf{alg}}^\star = (\sigma_1^\star,\ldots,\sigma_5^\star)$.
Let $n^\star := F(M^\star)$ and $r^\star := H(M^\star,n^\star)$.
Since verification accepts, the algebraic predicate
\(
  V_0^\star = \sigma_1^\star\sigma_2^\star - \sigma_5^\star,\), \(
  V_1^\star = \sigma_1^\star\sigma_2^\star - \sigma_3^\star + r^\star\sigma_4^\star,
\)
satisfies $\mathsf{SSS}_{2,2}^{-1}(V_0^\star,V_1^\star)=0$.
We distinguish two cases.

\noindent\emph{Case 1: $\mathcal{A}$ never queried $H$ on $(M^\star,n^\star)$.}
Then $r^\star$ is uniform and independent of $\mathcal{A}$'s view. For any fixed
$(\sigma_1^\star,\ldots,\sigma_5^\star)$, the algebraic predicate holds for at
most one value of $r^\star$, so
\(
  \Pr[\mathsf{Succ}_2 \wedge \text{Case 1}] \le \frac{1}{p}.
\)

\noindent\emph{Case 2: $\mathcal{A}$ queried $H$ on $(M^\star,n^\star)$.}
Let this be the first such query. At that moment, the simulator chooses
$r^\star$ uniformly. Conditioned on this choice, the distribution of
$(\sigma_1^\star,\ldots,\sigma_5^\star)$ is identical to that in the algebraic
core experiment. Define $\mathcal{B}_1$ as the algebraic-core adversary that
runs $\mathcal{A}$ and outputs the forgery in Case~2. Then
\(
  \Pr[\mathsf{Succ}_2 \wedge \text{Case 2}]
  \;\leq\;
  \mathsf{Adv}^{\mathsf{core}}(\mathcal{B}_1).
\)

If $\mathcal{A}$ ever forces $H(x)=H(x')$ for $x\neq x'$, then we obtain a
collision in $H$. Define $\mathcal{B}_3$ to output such a pair. Then
\(
  \Pr[\text{$\mathcal{A}$ causes a collision in $H$}]
  \;\leq\;
  \mathsf{Adv}^{\mathsf{CR}}_{H}(\mathcal{B}_3).
\)
Combining the two cases,
\(
  \Pr[\mathsf{Succ}_2]
  \;\leq\;
  \mathsf{Adv}^{\mathsf{core}}(\mathcal{B}_1)
  + \frac{1}{p}
  + \mathsf{Adv}^{\mathsf{CR}}_{H}(\mathcal{B}_3).
\)
Together with the transition from Game~0 to Game~1,
\[
  \mathsf{Adv}^{\mathsf{EUF\text{-}CMA}^{\neg\mathsf{DV}}}_{\mathsf{SILMARILS}}(\mathcal{A})
  \;\leq\;
  \mathsf{Adv}^{\mathsf{PRF}}_{\mathrm{HMAC}}(\mathcal{B}_2)
  + \mathsf{Adv}^{\mathsf{core}}(\mathcal{B}_1)
  + \mathsf{Adv}^{\mathsf{CR}}_{H}(\mathcal{B}_3)
  + \frac{1}{p}.
\]
By Lemma~\ref{algcore} and the assumed security of HMAC and $H$, the right-hand
side is negligible. Hence \textsc{SILMARILS} is
$\mathsf{EUF\text{-}CMA}^{\neg\mathsf{DV}}$ secure.
\hfill$\square$
\end{proof}
\subsection{$\mathsf{EUF\text{-}CMA}^{\neg\mathsf{DV}}$ Security of Two-Party \textsc{SILMARILS} in the QROM}
\label{subsec:qrom-security}
We now extend the $\mathsf{EUF\text{-}CMA}^{\neg\mathsf{DV}}$ analysis of \textsc{SILMARILS} to the QROM. The classical ROM proof relies on lazy
sampling and post-hoc programming of the random oracle at points
$(M,n)$, where $n$ is the output of a PRF. In the QROM, the adversary
may query the random oracle in superposition, and therefore the
classical programming argument is no longer sound. To obtain a valid
security proof, we rely on the \emph{measure-and-reprogram} technique of
Zhandry~\cite{Zhandry2012} and Unruh~\cite{Unruh2015}, which allows
programming the random oracle at a single point while bounding the
adversary's distinguishing advantage.
\begin{lemma}[Quantum Measure-and-Reprogram \cite{Zhandry2012}, \cite{Unruh2015}]\label{measurereplemma}
Let $\mathcal{A}$ be a quantum algorithm that makes at most $q_H$ queries
to a random oracle $H : \mathcal{X} \to \mathcal{Y}$, where
$\mathcal{Y}$ is a finite set. On input public parameters $\mathsf{PP}$,
$\mathcal{A}^H$ outputs a pair $(x^\star,z) \in \mathcal{X} \times \mathcal{Z}$.
Let $\mathsf{V}(\mathsf{PP},x^\star,z,H(x^\star)) \in \{0,1\}$ be a (classical) predicate.
Consider the following two experiments:
\begin{description}
  \item[Real experiment ($E_1$):]
    Sample a random oracle $H : \mathcal{X} \to \mathcal{Y}$, run
    $(x^\star,z) \leftarrow \mathcal{A}^H(\mathsf{PP})$, and output
    $(x^\star,z,H(x^\star))$.

  \item[Measure-and-reprogram experiment ($E_2$):]
    Sample a random oracle $H$ and run $\mathcal{A}^H(\mathsf{PP})$, but let a
    simulator $\mathcal{S}$ do the following: at one adaptively chosen oracle
    query (among at most $q_H$), measure the query register to obtain
    $\tilde{x}$; choose any $y^\star \in \mathcal{Y}$; define
    $H'(\tilde{x}) = y^\star$ and $H'(x)=H(x)$ for $x\neq\tilde{x}$; and let
    $\mathcal{A}$ continue with $H'$.
    If $(x^\star,z)$ is the final output, return $(x^\star,z,H'(x^\star))$.
\end{description}
Then a simulator $\mathcal{S}$ exists making at most $q_H + O(1)$ oracle queries and
\begin{eqnarray*}
  \Big|
  \Pr\big[ \mathsf{V}(\mathsf{PP},x^\star,z,H(x^\star)) = 1 \big| E_1 \big]
  -
  \Pr\big[ \mathsf{V}(\mathsf{PP},x^\star,z,H'(x^\star)) = 1 \big| E_2 \big]
\Big|
  \leq  O\!\left(\frac{q_H^2}{|\mathcal{Y}|}\right).
\end{eqnarray*}
In particular, from the adversary's point of view, programming the
random oracle at a single point $\tilde{x}$ (chosen by measuring one of
its queries) and assigning it an arbitrary value $y^\star$ is
indistinguishable from the original random-oracle experiment up to
statistical distance $O(q_H^2/|\mathcal{Y}|)$.
\end{lemma}
\begin{theorem}[$\mathsf{EUF\text{-}CMA}^{\neg\mathsf{DV}}$ Security in the QROM]\label{thm:silmarils-qrom}
Assume $H$ is modeled as a quantum-accessible random oracle and is
        collision-resistant against quantum adversaries and $\mathrm{HMAC}$ is a quantum-secure PRF keyed by $k_{\mathsf{sig}}$.
Then, \textsc{SILMARILS} is $\mathsf{EUF\text{-}CMA}^{\neg\mathsf{DV}}$ secure in the QROM. More precisely,
for any quantum polynomial time (QPT) adversary $\mathcal{A}$ making at most $q_H$ quantum queries
to $H$, there exist QPT algorithms $\mathcal{B}_1,\mathcal{B}_2,
\mathcal{B}_3$ such that
\[
  \mathsf{Adv}^{\mathsf{EUF\text{-}CMA^{\neg\mathsf{DV}}}}_{\mathsf{SILMARILS}}(\mathcal{A})
  \;\leq\;
  \mathsf{Adv}^{\mathsf{core}}(\mathcal{B}_1)
  + \mathsf{Adv}^{\mathsf{PRF}}_{\mathrm{HMAC}}(\mathcal{B}_2)
  + \mathsf{Adv}^{\mathsf{CR}}_{H}(\mathcal{B}_3)
  + O\!\left(\frac{q_H^2}{p}\right).
\]
\end{theorem}

\begin{proof}
Let $\mathcal{A}$ be a QPT adversary against \textsc{SILMARILS} in the
$\mathsf{EUF\text{-}CMA}^{\neg\mathsf{DV}}$ game, making at most $q_H$ quantum queries to the random oracle
$H$. We show how to bound
$\mathsf{Adv}^{\mathsf{EUF\text{-}CMA^{\neg\mathsf{DV}}}}_{\mathsf{SILMARILS}}(\mathcal{A})$
via a sequence of games and reductions.

\noindent\textbf{Game 0 (Real QROM $\mathsf{EUF\text{-}CMA}^{\neg\mathsf{DV}}$ game):}
In Game~0, we run the standard $\mathsf{EUF\text{-}CMA}^{\neg\mathsf{DV}}$ experiment for
\textsc{SILMARILS} in the QROM: 1) A key pair $(\mathsf{pk},\mathsf{sk})$ is generated honestly. 2) The adversary $\mathcal{A}$ has quantum oracle access to $H$ and  classical access to a signing oracle $\mathsf{Sign}_{\mathsf{sk}}(\cdot)$. 3) Eventually, $\mathcal{A}$ outputs a pair $(M^\star,\sigma_{\mathsf{alg}}^\star)$. 4) The experiment outputs $1$ (success) if
        $\mathsf{Verify}_{\mathsf{pk}}(M^\star,\sigma_{\mathsf{alg}}^\star)=1$ and
        $M^\star$ was never queried to the signing oracle.
Let
\(
  \mathsf{Adv}_0
  := \Pr[\text{Game 0 outputs }1]
  = \mathsf{Adv}^{\mathsf{EUF\text{-}CMA^{\neg\mathsf{DV}}}}_{\mathsf{SILMARILS}}(\mathcal{A}).
\)

\noindent\textbf{Game 1 (Replace HMAC by a quantum-secure PRF):}
In the real scheme, the value $n$ used to derive $r$ is computed as
$n = \mathrm{HMAC}_{k_{\mathsf{sig}}}(M)$. By assumption,
$\mathrm{HMAC}$ is a quantum-secure PRF. We define Game~1 as follows: 1) Replace $\mathrm{HMAC}_{k_{\mathsf{sig}}}$ by a truly random
        function $F : \{0,1\}^\ast \to \mathbb{F}_p$. 2) For each signing query $M$, set $n := F(M)$ and then derive
        $r := H(M,n)$ as in the real scheme.
By the quantum PRF security of HMAC, there exists a QPT distinguisher
$\mathcal{B}_2$ such that
\(
  |\mathsf{Adv}_1 - \mathsf{Adv}_0|
  \;\leq\;
  \mathsf{Adv}^{\mathsf{PRF}}_{\mathrm{HMAC}}(\mathcal{B}_2),
\)
where $\mathsf{Adv}_1$ is the success probability of $\mathcal{A}$ in
Game~1. Hence it suffices to bound $\mathsf{Adv}_1$ and consider the idealized setting where $F$ is a truly
random function, independent of $H$.

\noindent\textbf{Game 2 (Measure-and-reprogram $H$ at signing points).}
In Game~1, for each signing query $M$, the scheme computes
$n := F(M)$ and then sets $r := H(M,n)$. In the QROM, $\mathcal{A}$ may
have queried $H$ on superpositions that include $(M,n)$, so we cannot
naively program $H(M,n)$ after the fact.
We now define Game~2, where a simulator $\mathcal{S}$ uses the
measure-and-reprogram lemma (Lemma~\ref{measurereplemma}) to
program $H$ at points of the form $(M,n)$ for signing queries $M$: 1) For each signing query $M$, the simulator samples
        $n := F(M)$ and a fresh $r \xleftarrow{\$} \mathbb{F}_p$. 2) The simulator applies the measure-and-reprogram procedure to the
        quantum random oracle $H$ at the point $(M,n)$, reprogramming    $H(M,n)$ to $r$. 3) It then computes the algebraic signature components
        $(\sigma_1,\ldots,\sigma_5)$ exactly as in the real signing
        algorithm, using this $r$, and returns
        $\sigma_{\mathsf{alg}} = (\sigma_1,\ldots,\sigma_5)$ to $\mathcal{A}$.
By Lemma~\ref{measurereplemma}, for each programmed point
$(M,n)$, the statistical distance between the real QROM experiment and
the measure-and-reprogram experiment is at most
$O(q_H^2/|\mathcal{Y}|)$, where $\mathcal{Y}$ is the output space of
$H$. In our setting, $H$ outputs elements of $\mathbb{F}_p$, so
$|\mathcal{Y}| = p$. Since the number of signing queries is polynomial
and $q_H$ bounds the total number of quantum queries, a union bound
yields
\(
  |\mathsf{Adv}_2 - \mathsf{Adv}_1|
  \;\leq\;
  O\!\left(\frac{q_H^2}{p}\right),
\)
where $\mathsf{Adv}_2$ is the success probability of $\mathcal{A}$ in
Game~2.
In Game~2, the distribution of signatures seen by $\mathcal{A}$ is
identical to the real scheme: for each signing query $M$, the value
$r = H(M,n)$ is uniform in $\mathbb{F}_p$ (from $\mathcal{A}$'s point of
view) and the signature is computed honestly from $r$.

\noindent\textbf{Game 3 (Explicit simulation of signatures):}
We now make explicit the fact that in Game~2, the simulator chooses
$r \xleftarrow{\$} \mathbb{F}_p$ and then programs $H(M,n)$ to $r$.
Thus, we can equivalently define Game~3 as: 1) For each signing query $M$, sample $n := F(M)$ and
        $r \xleftarrow{\$} \mathbb{F}_p$. 2) Program $H(M,n) := r$ via measure-and-reprogram. 3) Compute $(\sigma_1,\ldots,\sigma_5)$ exactly as in the real        signing algorithm using this $r$.
By construction, $\mathsf{Adv}_3 = \mathsf{Adv}_2$. We emphasize that
in Game~3, for each signed message $M$, the pair $(n,r)$ is uniform and
independent of $\mathcal{A}$'s view, conditioned only on the fact that
$H(M,n)$ has been programmed to $r$.

Let $(M^\star,\sigma_{\mathsf{alg}}^\star)$ be the forgery output by $\mathcal{A}$ in
Game~3, where
$\sigma_{\mathsf{alg}}^\star = (\sigma_1^\star,\ldots,\sigma_5^\star)$ and
$\mathsf{Verify}_{\mathsf{pk}}(M^\star,\sigma_{\mathsf{alg}}^\star)=1$. Let
$n^\star := F(M^\star)$ and $r^\star := H(M^\star,n^\star)$.
We distinguish two cases, depending on whether $\mathcal{A}$ has
queried $H$ on $(M^\star,n^\star)$.
\emph{Case 1: $\mathcal{A}$ never queried $H$ on $(M^\star,n^\star)$.}
In this case, from $\mathcal{A}$'s point of view, $r^\star$ is uniform
in $\mathbb{F}_p$ and independent of its entire view (including all
signatures and oracle answers), because: 1) $F$ is a truly random function, so $n^\star = F(M^\star)$ is
        uniform and independent of all other values; 2) $H$ is a random oracle, and $(M^\star,n^\star)$ was never
        queried, so $H(M^\star,n^\star)$ is uniform and independent.
The algebraic verification predicate of \textsc{SILMARILS} is a system
of equations over $\mathbb{F}_p$ that, for fixed
$(\sigma_1^\star,\ldots,\sigma_5^\star)$, is satisfied by at most one
value of $r^\star$. Therefore,
\(
  \Pr[\text{verification accepts in Case 1}]
  \;\leq\;
  \frac{1}{p}.
\)
\emph{Case 2: $\mathcal{A}$ queried $H$ on $(M^\star,n^\star)$.}
Let $(M^\star,n^\star)$ be the first such query. At the moment this
query is made, the simulator in Game~3 chooses $r^\star$ as a fresh
uniform element of $\mathbb{F}_p$ and programs $H(M^\star,n^\star)$ to
$r^\star$ via measure-and-reprogram. By
Lemma~\ref{measurereplemma}, the joint state of $\mathcal{A}$ and
the oracle after this programming step is statistically close to the
state in the ideal random-oracle experiment, with distinguishing
advantage at most $O(q_H^2/p)$.
Conditioned on this, the distribution of
$(\sigma_1^\star,\ldots,\sigma_5^\star,r^\star)$ in $\mathcal{A}$'s
view is identical to that in the algebraic core experiment, where
$r^\star$ is uniform and independent of the adversary's internal
randomness. Hence, any non-negligible probability that
$\mathcal{A}$ outputs a new tuple
$(M^\star,\sigma_1^\star,\ldots,\sigma_5^\star)$ such that the algebraic
predicate holds for $r^\star$ yields a non-negligible advantage for a
QPT adversary $\mathcal{B}_1$ in the algebraic core game. By Lemma~\ref{algcore}, this probability is bounded by
$\mathsf{Adv}^{\mathsf{core}}(\mathcal{B}_1) + 1/p$.

In both cases, we have implicitly assumed that $\mathcal{A}$ does not
exploit collisions in $H$. If $\mathcal{A}$ can force
$H(x) = H(x')$ for distinct $x \neq x'$ with non-negligible probability,
then we obtain a quantum collision-finder $\mathcal{B}_3$ against $H$,
with advantage $\mathsf{Adv}^{\mathsf{CR}}_H(\mathcal{B}_3)$.
Now, let $\mathsf{Adv}_3$ be the success probability of $\mathcal{A}$ in
Game~3. From the case analysis above, we obtain
\(
  \mathsf{Adv}_3
  \;\leq\;
  \mathsf{Adv}^{\mathsf{core}}(\mathcal{B}_1)
  + \mathsf{Adv}^{\mathsf{CR}}_H(\mathcal{B}_3)
  + \frac{1}{p}
  + O\!\left(\frac{q_H^2}{p}\right).
\)
Combining the game transitions, we have
\begin{eqnarray*}
\tiny
  \mathsf{Adv}^{\mathsf{EUF\text{-}CMA^{\neg\mathsf{DV}}}}_{\mathsf{SILMARILS}}(\mathcal{A})&=&\mathsf{Adv}_0 \leq \mathsf{Adv}_1
  + \mathsf{Adv}^{\mathsf{PRF}}_{\mathrm{HMAC}}(\mathcal{B}_2)
   \leq  \mathsf{Adv}_3
  + \mathsf{Adv}^{\mathsf{PRF}}_{\mathrm{HMAC}}(\mathcal{B}_2)
  + \\O\!\left(\frac{q_H^2}{p}\right)
   &\leq& \mathsf{Adv}^{\mathsf{core}}(\mathcal{B}_1)
  + \mathsf{Adv}^{\mathsf{PRF}}_{\mathrm{HMAC}}(\mathcal{B}_2)
  + \mathsf{Adv}^{\mathsf{CR}}_{H}(\mathcal{B}_3)
  + O\!\left(\frac{q_H^2}{p}\right).
\end{eqnarray*}
By Lemma~\ref{algcore} and the assumed quantum security of HMAC
and $H$, the right-hand side is negligible in the security parameter.
\hfill$\square$
\end{proof}
\section{Security Analysis of Three-Party \textsc{SILMARILS}}\label{Security_sec2}
We now state the main security theorem for our three-party
\textsc{SILMARILS} protocol. The theorem asserts that the protocol
implements the ideal three-party digital signature functionality
$f_{\mathsf{DS}}$ with error $\varepsilon = 1/p$.
\begin{theorem}[Security of Three-Party \textsc{SILMARILS}]\label{thm:3p-silmarils-unified}
Let $p$ be a prime and let all parties operate over $\mathbb{F}_p$.
Let $H$ be any function (deterministic, classical random oracle, or
quantum-accessible random oracle) with range $\mathbb{F}_p$, and define
\(
  x := H(M,\sigma_{\mathsf{alg}}).
\)
Consider the three-party \textsc{SILMARILS} protocol consisting of
Algorithms~\ref{sign}--\ref{alg:3p-verifytransfer}.
Then, for any (classical or quantum) adversary corrupting at most one
party, the protocol realizes the ideal digital-signature functionality
$f_{\mathsf{DS}}$ with distinguishing advantage
\(
  \varepsilon = \frac{1}{p}.
\)
In particular, the protocol satisfies the previously defined notions of
correctness, unforgeability, transferability, and secrecy, each with
error at most $1/p$, uniformly across the ROM, QROM, and Pure-IT models.
\end{theorem}
\begin{proof}
Fix a prime $p$ and let all parties operate over $\mathbb{F}_p$. Let
$H$ be any function with range $\mathbb{F}_p$ (deterministic, classical
random oracle, or quantum-accessible random oracle). For a given message
$M$ and honestly generated algebraic signature
$\sigma_{\mathsf{alg}} = (\sigma_1,\ldots,\sigma_5)$ produced by
Algorithm~\ref{sign}, the signer $\mathcal{P}_1$ computes
\(
  x := H(M,\sigma_{\mathsf{alg}}) \in \mathbb{F}_p
\)
and then runs Algorithms~\ref{alg:3p-sign}--\ref{alg:3p-ic-checks} to
authenticate $x$ using an IC protocol of the same
form as Protocol~18.7 in \cite{IT-book}.

\noindent\textbf{Reduction to the IC functionality:}
Condition on a fixed value of $x \in \mathbb{F}_p$ and on fixed secret
coefficients $(k_1,k_2) \in \mathbb{F}_p^2$ held by $\mathcal{P}_3$.
From the point of view of the IC layer, the three-party
\textsc{SILMARILS} protocol is exactly an execution of Protocol~18.7 in \cite{IT-book} on
the message $x$, where $\mathcal{P}_1$ plays the role of the sender, holding $x$ and
        the affine authentication pair $(x,\sigma)$ with
        $\sigma = k_1 x + k_2$; $\mathcal{P}_2$ plays the role of the receiver, who should
        obtain $z_2 = x$; $\mathcal{P}_3$ plays the role of the verifier, who holds
        $(k_1,k_2)$ and outputs $z_3 \in \{x,\bot\}$.
Algorithms~\ref{alg:3p-sign}--\ref{alg:3p-ic-checks} implement exactly
the challenge, consistency checks, and resolution rules of
Protocol~18.7 in \cite{IT-book}, specialized to the affine line
$\sigma = k_1 x + k_2$ over $\mathbb{F}_p$.
Crucially, the IC security properties (correctness, unforgeability,
transferability, secrecy) depend only on the algebraic structure of the
protocol over $\mathbb{F}_p$ and on the randomness of the IC coins
($k_1,k_2,x',k_2',e$), and \emph{not} on how $x$ was obtained. In
particular, the adversary never interacts with $H$ through the IC
interface: $x$ is simply a field element fixed by $\mathcal{P}_1$ before
the IC protocol starts. Therefore, the analysis of \cite[Protocol~18.7]{IT-book}   applies verbatim in all three models (ROM, QROM, Pure-IT).

\noindent\textbf{Correctness:}
If $\mathcal{P}_1$ and $\mathcal{P}_2$ are honest, then by the
correctness guarantee of Protocol~18.7, the IC protocol ensures that
$\mathcal{P}_2$ and $\mathcal{P}_3$ both output $x$ except with
probability at most $1/p$. In our instantiation, the resolution rules in
Algorithm~\ref{alg:3p-ic-checks} guarantee that whenever the IC protocol
accepts, we have $z_2 = x$ and $z_3 = x$. Thus the previously defined
correctness notion holds with error at most $1/p$.

\noindent\textbf{Unforgeability:}
If $\mathcal{P}_1$ and $\mathcal{P}_3$ are honest, then a malicious
$\mathcal{P}_2$ attempts to cause $\mathcal{P}_3$ to output some
$z_3 \notin \{x_1,\bot\}$, where $x_1$ is the value authenticated by
$\mathcal{P}_1$. In the IC abstraction, this corresponds to the receiver
trying to change the authenticated message while still passing the
verifier's affine check. By the soundness (unforgeability) guarantee of
Protocol~18.7 and Lemma~18.15 in \cite{IT-book}, the probability that a cheating receiver
can make the verifier accept a value different from the sender's input
is at most $1/p$. Hence the previously defined unforgeability condition
holds with error at most $1/p$.

\noindent\textbf{Transferability:}
If $\mathcal{P}_2$ and $\mathcal{P}_3$ are honest, then a malicious
$\mathcal{P}_1$ may try to make them disagree on the authenticated
value. In the IC abstraction, this is exactly the cheating-sender case:
the sender attempts to cause the receiver and verifier to output
different values. Protocol~18.7 guarantees that this happens with
probability at most $1/p$. Therefore, the previously defined
transferability notion holds with error at most $1/p$.

\noindent\textbf{Secrecy:}
If $\mathcal{P}_1$ and $\mathcal{P}_2$ are honest, then the view of
$\mathcal{P}_3$ during the IC protocol consists only of the affine
authentication data and the challenge transcript, which, by the privacy
guarantee of Protocol~18.7, is statistically independent of the
authenticated value $x_1$. Thus the previously defined secrecy notion
holds with error $0$ (and hence at most $1/p$).

In all of the above arguments, the value \(x\) is treated as an arbitrary element of \(\mathbb{F}_p\) fixed before the IC protocol begins. The correctness, unforgeability, transferability, and secrecy proofs for Protocol~18.7 and Lemma~18.15 in~\cite{IT-book} rely only on the fact that \(H\) outputs elements of \(\mathbb{F}_p\); they do not depend on any further structural property of \(H\). As a result, the same bounds apply whether \(H\) is instantiated as a deterministic function, a classical random oracle, or a quantum-accessible random oracle. The distinguishing advantage between the real three-party \textsc{SILMARILS} protocol and the ideal digital-signature functionality \(f_{\mathsf{DS}}\) is therefore at most \(1/p\) in all three models. This completes the proof.
\hfill$\square$
\end{proof}

\section{Performance Evaluation}\label{sec:performance}
We now evaluate the efficiency of two-party \textsc{SILMARILS}, covering key and
signature sizes, computational costs of signing and verification, and overall
communication overhead, and we compare these metrics to leading PQ signature
schemes to contextualize practical performance.  \emph{As \textsc{SILMARILS} in
two-party TDV mode is not a public-verifier signature scheme, these comparisons
to Dilithium, Falcon, and SPHINCS$^+$ reflect size and efficiency only, not
equivalent security or functionality.}

Let $p \approx 2^k$ denote the prime defining the field $\mathbb{F}_p$. All
public, secret, and signature components are single field elements, and thus
occupy $\log_2(p) \approx k$ bits. This yields exceptionally compact keys.
The  secret key consists of a single field element, giving
\(
|\mathsf{sk}| = k \text{ bits}
\).
The public key contains two field elements giving
\(
|\mathsf{pk}| = 2k \text{ bits}
\).
A signature consists of five field elements giving
\(
|\sigma| = 5k \text{ bits}
\).
Over a 256-bit prime field, we have $|\mathsf{sk}| = 256 \text{ bits} = 32 \text{B}$, $|\mathsf{pk}| = 512 \text{ bits} = 64 \text{B}$ and $|\sigma| = 1280 \text{ bits} = 160 \text{B}$.
Thus, even at the 256-bit security level, the public and secret keys remain
extremely compact, and the signature size remains close to that of ECDSA signature.
In addition, all arithmetic is carried out in $\mathbb{F}_p$ with $p \approx 2^{256}$, so a field multiplication costs roughly the same as a constant-time 256-bit modular multiply (e.g., via Montgomery reduction). Signing requires about $10$--$12$ such multiplications, while verification needs only $5$--$6$, making both operations far lighter than those in lattice or hash-based PQ schemes.
Table~\ref{tab:merged_comparison_enhanced} compares our concrete instantiation
to standardized PQ signatures at the 128- and 192-bit security levels. Since
\textsc{SILMARILS} is a TDV scheme with a different threat model, this comparison
is not meant to suggest equivalence, but to illustrate efficiency in blockchain
settings where PQC is more than sufficient but still expensive. A fair comparison with other DV
schemes is omitted due to space and the complexity of aligning models. Additional
benchmarks appear in \cite{silmarils-implementation}.
\begin{table}[tb]
\centering
\caption{Comparison of \textsc{SILMARILS} with other digital signature schemes.}\label{tab:merged_comparison_enhanced}
\resizebox{\textwidth}{!}{ 
\begin{tabular}{@{}l l S[table-format=4.0] S[table-format=4.0] S[table-format=5.0] l l@{}} 
\toprule
\textbf{Scheme} & \textbf{Security Type} & {\textbf{PK Size (B)}} & {\textbf{SK Size (B)}} & {\textbf{Sig Size (B)}} & \textbf{Who Can Verify and Cost} & \textbf{Key Assumptions} \\
\midrule
\rowcolor{lightgray}
\textsc{SILMARILS} & \textbf{Information-Theoretic} & 64 & 32 & 160 &\textbf{TDV}, $\approx 5-6$ field mults + 1 SSS recon & TRNG + SSS  \\
\addlinespace[0.5em] 

sr25519 & Classical, Like Schnorr & 32 & 64 & 64 &Public Verifier, 1 scalar mult + 1 hash & ECDLP over Curve2551 \\
\midrule
\multicolumn{7}{l}{\qquad\qquad\qquad\qquad\qquad\qquad\qquad\qquad\qquad\qquad\qquad\qquad\qquad\qquad\textbf{NIST Post-Quantum Candidates}} \\
\midrule
Dilithium-2 (L2) & Computational PQ & 1312 & 2528 & 2420 &Public Verifier, 1 matrix-vector mult + 1 hash & Module-LWE \\
Dilithium-3 (L3) & Computational PQ & 1952 & {4000\rlap{$^\text{*}$}} & 3293 &Public Verifier, matrix--vector mult & Module-LWE \\
Falcon-512 (L1) & Computational PQ & 897 & 1281 & 690 &Public Verifier, FFT + lattice decoding & SIS over NTRU lattices \\
SPHINCS$^+$ (Gen.) & Computational PQ & 64 & 128 & 49856 &Public Verifier, Many hashing + tree traversal & Hash-based (QROM) \\
SPHINCS$^+$-128s (L1) & Computational PQ & 32 & {64\rlap{$^\text{*}$}} & 7856 &Public Verifier, hash tree traversal & Hash-based \\
SPHINCS$^+$-128f (L1) & Computational PQ & 32 & {64\rlap{$^\text{*}$}} & 17088 &Public Verifier, hash tree traversal & Hash-based \\
SPHINCS$^+$-192s (L3) & Computational PQ & 48 & 96 & 16224 &Public Verifier, hash tree traversal & Hash-based \\
SPHINCS$^+$-192f (L3) & Computational PQ & 48 & 96 & 35664 &Public Verifier, hash tree traversal & Hash-based \\
SLH-DSA-(SHA2/SHAKE)-128-24\rlap{$^\text{**}$} \hspace{0.3cm}(L1, 224 sigs)  & Computational PQ & 32 & 48 & 3856 &Public Verifier, hash-based tree traversal & Hash-based (limited-signature) \\
SLH-DSA-(SHA2/SHAKE)-192-24\rlap{$^\text{**}$} \hspace{0.3cm}(L3, 224 sigs) & Computational PQ & 48 & 72 & 7752 &Public Verifier, hash-based tree traversal & Hash-based (limited-signature) \\
SLH-DSA-(SHA2/SHAKE)-256-24\rlap{$^\text{**}$} \hspace{0.3cm}(L5, 224 sigs)  & Computational PQ & 64 & 96 & 14944 &Public Verifier, hash-based tree traversal & Hash-based (limited-signature) \\
\bottomrule\\
\multicolumn{7}{l}{\footnotesize *SK Size for Dilithium-3 assumed to scale with security level (not explicitly in original tables). SK Size for SPHINCS$^{+}$ variants assumed to be 64B (seed size}\\
\multicolumn{7}{l}{\footnotesize for 128-bit security level,  not explicitly in original tables). See \url{https://asecuritysite.com/pqc/pqc_sig} for more details.} \\
\multicolumn{7}{l}{\footnotesize **NIST SP 800‑230 \cite{NIST-SP-800-230-ipd}: reduced-size SLH-DSA variants for $2^{24}$-signature use cases.}
\\
\end{tabular}
}
\vspace{-0.7cm}
\end{table}
\vspace{-0.4cm}
\section{Conclusion}\label{sec:conclusion}
\textsc{SILMARILS} is a unified framework for lightweight TDV authentication
built from a minimal algebraic core over $\mathbb{F}_p$ and perfect
$2$-out-of-$2$ Shamir sharing, rather than a conventional digital signature.
Its two-party mode yields a TDV signature scheme: we prove designated-verifier
simulatability in the sense of Jakobsson--Sako--Impagliazzo and establish
$\mathsf{EUF\text{-}CMA}^{\neg\mathsf{DV}}$ security for all non-designated
verifiers in both the ROM and QROM. In the three-party setting, adopting the
broadcast model of Fitzi~\emph{et~al.}, we obtain a statistically secure
information-checking protocol with full simulation-based security and error~$1/p$.
Practical instantiations rely on hash functions, motivating a unified analysis
across the Pure IT, IT+ROM, and QROM models and a quantum extension of the
Fitzi simulation framework. We show that correctness, secrecy, transferability,
and non-designated unforgeability remain equivalent to simulation-based security
even against quantum adversaries. Thanks to its structural simplicity,
\textsc{SILMARILS} achieves much smaller keys and signatures than standardized
PQC schemes. This is not a comparison of primitives—TDV authentication and PQC
signatures address different threat models—but in blockchain settings where PQC
is already more than sufficient yet expensive, lightweight TDV mechanisms such
as \textsc{SILMARILS} offer a compelling alternative for reducing on-chain
overhead.

\bibliographystyle{splncs04}
\bibliography{ReffSAC}
\end{document}